\documentclass[letterpaper,11pt]{article}

\usepackage{amsmath,amssymb,amsthm,bbm}
\usepackage{hyperref}
\hypersetup{colorlinks=false,pdfborder={0 0 0}}
\usepackage{color}
\usepackage{float}
\usepackage{graphicx}
\usepackage[small]{caption}
\usepackage[margin=1in]{geometry}
\usepackage{algorithm,algpseudocode}
\usepackage{parskip}

\usepackage[all]{xy}
\usepackage{young}

\usepackage{bbm}

\newcommand\E{\mathbb E}

\renewcommand\P{\mathbb P}
\renewcommand\L{\mathcal{L}}
\newcommand\X{\mathbb X}
\renewcommand\Pr{\mathbb P}

\newcommand\A{\mathcal A}
\newcommand\B{\mathcal B}
\newcommand\C{\mathcal C}
\newcommand\BC{\mathcal{BC}}

\newcommand\ignore[1]{ }

\DeclareMathOperator{\NB}{NB}

\DeclareMathOperator{\Bern}{Bernoulli}

\newtheorem{theorem}{Theorem}[section]
\newtheorem{lemma}[theorem]{Lemma}

\newtheorem{proposition}[theorem]{Proposition}
\theoremstyle{definition}
\newtheorem{definition}{Definition}[section]
\newtheorem{remark}[theorem]{Remark}
\newtheorem{example}[theorem]{Example}

\begingroup
    \makeatletter
    \@for\theoremstyle:=definition,remark,plain\do{%
        \expandafter\g@addto@macro\csname th@\theoremstyle\endcsname{%
            \addtolength\thm@preskip\parskip
            }%
        }
\endgroup

\def\Geo{\mbox{Geo}}
\def\Bern{\mbox{Bern}}

\begin{document}

\def\spacingset#1{\renewcommand{\baselinestretch}%
{#1}\small\normalsize} \spacingset{1}

\def\LS{LS}

\parskip 10pt

\begin{center}
\LARGE Random Sampling of Latin squares via binary contingency tables and probabilistic divide-and-conquer\\[.3cm]
\normalsize Stephen DeSalvo\footnote{Department of Mathematics, University of California Los Angeles. stephendesalvo@math.ucla.edu}\\[.1cm]
March 24, 2017\\[.5cm]
\parbox{.8\textwidth}{
\small
\textit{Abstract.}
We demonstrate a novel approach for the random sampling of Latin squares of order~$n$ via probabilistic divide-and-conquer. 
The algorithm divides the entries of the table modulo powers of $2$, and samples a corresponding binary contingency table at each level. 
The sampling distribution is based on the Boltzmann sampling heuristic, along with probabilistic divide-and-conquer. 
}
\end{center}
{
\ignore{
\small
\smallskip
\noindent \textbf{Keywords.} Random Sampling, Latin Square, Sudoku, Probabilistic Divide-and-Conquer, Rejection Sampling

\noindent \textbf{MSC classes:} 60C05, 65C50, 60-04 \\
\noindent \textbf{CCS 2012:} Mathematics of computing$\sim$Probabilistic algorithms, Theory of computation$\sim$Generating random combinatorial structures
}
}

\section{Introduction}

\subsection{Motivation}

When tasked with random sampling from a complicated state space, a popular choice is to run a rapidly mixing Markov chain. 
Markov chains have proven to be an incredibly powerful and versatile tool, however, 
there are several reasons why one would consider alternative methods. 
In particular, unless the chain starts off in stationarity, or has been coupled from the past~\cite{ProppWilson}, the error is seldom zero after a finite number of steps; quite often this is an acceptable tradeoff, and through a more complicated analysis one can bound the sampling error.
For sampling problems with a large number of constraints, however, it is often difficult to analyze the resulting Markov chain and prove mixing time bounds. 
In addition, there is no general scheme to parallelize the sequential steps in a Markov chain, and many naive approaches can lead to unintended sampling bias~\cite[Chapter~18]{wilkinson2006parallel}.

Another paradigm of random sampling, popular in combinatorics, is the Boltzmann sampler~\cite{Boltzmann, Flajolet}, where one samples from some combinatorial structure via independent random variables, determined by the form of the generating function to closely resemble the true distribution of component sizes in the random combinatorial structure. 
The result is a random combinatorial object with a random set of statistics, referred to as a \emph{grand canonical ensemble} in the statistical mechanics literature. 
For a fixed set of conditions, an exact sample can be obtained by sampling repeatedly until all of the conditions are satisfied, throwing away samples which miss the target~\cite{Duchon:2011aa}. 
We have the structure which is typical of Boltzmann sampling, which is a collection of independent random variables subject to a condition. 
However, unlike the plethora of combinatorial structures for which a randomly generated structure with random statistics contains an acceptable amount of bias, there are good reasons to demand that all statistics are satisfied, see for example~\cite{chen, PittelSetPartitions}. 
Fortunately, owing to the large amount of independence in the formulation, Boltzmann samplers are embarrassingly parallel, offering many different effective means of parallelizing the computation. 
Unfortunately, as with many multivariate combinatorial structures, the rejection costs associated with waiting until all statistics match exactly are prohibitively large, even with effective parallelizing. 

The Boltzmann sampler has been described as a generalization of the table methods of Nijenhuis and Wilf, also known as the \emph{recursive method}, see~\cite{NW}, although from our point of view it has a distinctly different flavor. 
The recursive method champions creating a table of numerical values which can be used to generate individual components of a sample in its correct proportion in an unbiased sample. 
This approach, however, demands the computing of a lookup table, which can be both prohibitively large in size and prohibitively long to compute each entry, typically by a recursion. 

Our approach lies somewhere in between the recursive method and Boltzmann sampling, with a distinctly probabilistic twist, and which avoids Markov chains altogether. 
Our approach, \emph{probabilistic divide-and-conquer} (PDC)~\cite{PDC}, provides an object which satisfies all constraints, and which targets the marginal distribution of components, in this case entries in a table, according to the Boltzmann sampling principle.  
Rather than sample all entries of the table at once and apply a rejection function which is either 0 or 1, we instead sample one bit of each entry in the table, one at a time, essentially building the table via its bits, starting with the least significant bit. 
In the case where rejection sampling probabilities are known or can be computed to arbitrary accuracy, the algorithm is unbiased. 
Probabilistic divide-and-conquer is similar in many respects to the recursive method, by approaching the problem in pieces and selecting a piece in proportion to already observed pieces. 
However, whereas the recursive method typically involves a more geometric/spatial decomposition of a finite, discrete sample space, we instead champion a probabilistic decomposition of the sample space, by decomposing the random variables themselves which describe random component-sizes. 
This approach also generalizes in a straightforward manner to sampling from continuous sample spaces and lower dimensional subspaces, see~\cite{PDCDSH}, as long as the random variables can be decomposed effectively. 

In the remainder of this section, we define $(r,c)$-contingency tables and Latin squares of order~$n$, and describe some of the standard algorithms to randomly sample them. 
Section~\ref{novel_algorithm} contains the highest-level explanation of our proposed algorithm for Latin squares of order~$n$, and in Section~\ref{section_probabilistic} we present the probabilistic tools which justify this approach. 
Section~\ref{algorithms} contains the complete statements of the algorithms, along with a word of caution in the final subsection. 

\subsection{$(r,c)$-contingency tables}


\begin{definition}
An \emph{$(r,c)$-contingency table} is an $m$ by $n$ table of nonnegative integer values with row sums given by $r = (r_1, \ldots, r_m)$ and column sums given by $c = (c_1, \ldots, c_n)$. 
A table for which the entries are further assumed to be in the set $\{0,1\}$ is called a \emph{binary $(r,c)$-contingency table}.
\end{definition}

The exact, uniform sampling from the set of $(r,c)$-contingency tables is a well-studied problem. 
Owing to the large parameter space, there are a plethora of results spanning many decades pertaining just to counting the number of such tables, see for example~\cite{Barvinok, BenderTables, Soules, BaldoniSilva, DeLoeraRational, BarvinokInequalities, BarvinokPermanents, GreenhillMcKay, BarvinokIntegerFlows, BarvinokApproximate, BarvinokHartigan, DiaconisGangolli, bender1978asymptotic, ONeil, GoodCrook}. 
There is also interest in various Markov chain approaches, including coupling from the past, see for example~\cite{chen, cryan2003polynomial, cryan2006rapidly, diaconissturmfels, fishman2012counting, dyergreenhill,kitajimamatsui, JerrumSinclair, Brualdi2007}. 

The random sampling of contingency tables is an extensive topic, and we do not attempt to recount all previously known results in this area, and instead refer the interested reader to, e.g., the introduction in~\cite{cryan2006rapidly} and the references therein. 
As mentioned previously, the main sampling approach traditionally championed is to use a Markov chain, e.g., starting with the general framework in~\cite{diaconissturmfels}, to define an appropriate state transition matrix, and then prove that such a chain is rapidly mixing; indeed, this approach has been profoundly fruitful, see for example~\cite{dyergreenhill, kitajimamatsui, JerrumSinclair}. 
Explicitly, one state transition championed in~\cite{DiaconisGangolli} is to sample two ordered rows and two ordered columns uniformly at random, and, \emph{if possible}, apply the transformation to their entries 
\[ \left(\begin{array}{cc} a_{11} & a_{12} \\ a_{21} & a_{22} \end{array}\right) \to \left(\begin{array}{cc} a_{11}+1 & a_{12} - 1 \\ a_{21} -1 & a_{22} + 1 \end{array}\right). \]
If such a transformation would force the table to lie outside the solution set, then a different pair of ordered rows and columns is generated. 
This chain was shown in~\cite{DiaconisGangolli} to be ergodic and converge to the uniform distribution over the set of contingency tables. 
Mixing times were later proved in~\cite{diaconis1993comparison, hernek1998random, chung1996sampling} in various contexts. 
Other Markov chains have also been proposed, see for example~\cite{dyergreenhill, cryan2006rapidly}. 
The recent book~\cite{huber2015perfect} also contains many more examples involving Markov chains and coupling from the past to obtain exact samples. 

An arguably unique approach to random sampling of contingency tables is contained in~\cite{dyer1997sampling}. The algorithm associates to each state a parallelepiped with respect to some basis, and defines some convex set which contains all of the parallelepipeds. 
Then, one samples from this convex set, and if the sample generated lies within one of the parallelepipeds, it returns the corresponding state; otherwise, restart. 
This algorithm was shown to be particularly effective  in~\cite{morris2002improved} when the row sums are all $\Omega(n^{3/2} m\log m)$ and the column sums are all $\Omega(m^{3/2} n\log n)$.

Recently, a self-similar PDC algorithm was championed for the random sampling of contingency tables~\cite{DeSalvoCT}, offering an arguably new approach to random sampling of these intricate structures. 
In related work, the author demonstrated how to utilize PDC to improve upon existing algorithms for exact random sampling of Latin squares of order~$n$~\cite{DeSalvoSudoku}. 
The current work, motivated by improving further still the random sampling of Latin squares of order $n$, extends the original PDC sampling algorithm for contingency tables to the more general case when certain entries are forced to be 0, see for example~\cite{BrualdiDahl, bezakova}. 
The probabilistic analysis is straightforward, although the computational complexity changes drastically. 
Nevertheless, our numerical experiments demonstrate that this approach is simple, practical and executes fairly rapidly for a large class of tables. 
Thus, we champion our approach for practitioners looking for simple, alternative methods for sampling quickly from these intricate structures. 


\subsection{Latin squares of order~$n$}
\label{novel_algorithm}

\begin{definition}
A Latin square of order $n$ is an $n \times n$ table of values such that the numbers from the set $\{1,2,\ldots, n\}$ appear exactly once in each row and in each column. 
\end{definition}

There are many techniques available for random sampling of Latin squares of order~$n$, see for example~\cite[Section~6]{DeSalvoSudoku}. 
A Markov chain approach is contained in~\cite{JacobsonMatthews}; in particular, they construct and analyze two different Markov chains. 
The transition states are described via an explicit algorithm, requiring $O(n)$ time to transition from one state to the next, and guaranteed to satisfy all of the Latin square constraints. 
However, as far as we are aware, even though the stationary distribution was proved to be uniform over the set of Latin squares of order~$n$, neither of the Markov chains presented have been shown to be rapidly mixing. 
Another approach is to decompose the entries of a Latin square as sums of mutually disjoint permutation matrices, as in~\cite{Dahl, Fontana, FontanaFractions, YordzhevNumber}. 
This approach has practical limitations, as the straightforward rejection algorithm is prohibitive, and the more involved approach in~\cite{Fontana} requires an auxiliary computation which dominates the cost of the algorithm. 
We suspect that this approach may benefit from a probabilistic divide-and-conquer approach, though we have not carried out the relevant analysis. 

An informal description of our algorithm for random sampling of Latin squares of order~$n$ is as follows; see Algorithm~\ref{latin:square:algorithm} for the complete description.  For simplicity of exposition, we assume $n = 2^m$ for some positive integer $m$, which has no theoretical significance, but avoids cumbersome notation involving rounding. 
First, sample an $n \times n$ binary contingency table with all row sums and column sums equal to $n/2$. 
Then consider the subset of entries in which a 1 appears, of which there are $n/2$ in each row and column, and sample within that subset of entries an $(n/2) \times (n/2)$ binary contingency table with all row sums and column sums equal to $n/4$; do the same for the subset of entries in which a 0 appears.
Repeat this process through all $m$ levels. 
By interpreting the sequence of $m$ 1s and 0s as the binary expansion of a positive integer, we obtain a Latin square of order~$n$. 

This idea, that of generating a Latin square by its bits, could be considered a straightforward divide-and-conquer algorithm if one simply desires an object from the solution set. 
What makes our treatment more intricate is our approximation heuristic, designed to target the uniform distribution over the solution set; see Section~\ref{latin_square_heuristic}. 

Of course, in order to avoid bias in the sampling algorithm, we must sample each of the $n/2^\ell \times n/2^\ell$ binary contingency tables in their correct proportion of partially completed Latin squares of order~$n$; this is not such a trivial task, and requires either counting all possible completions, a computationally extensive task, or a more direct analysis of the resulting multivariate probability governing the acceptance probability, given explicitly in Equation~\eqref{fij}.  While both approaches present technical difficulties, the probabilistic formulation yields a natural approximation heuristic, which we present in Equation~\eqref{alternative:rejection}, which is computable in polynomial time.  It assumes independence of columns, a reasonable approximation heuristic in many parameter spaces of interest, while still enforcing several necessary conditions. 
A more complete probabilistic description of our heuristic is contained in Section~\ref{latin_square_heuristic}.

\begin{example}
The following is an example of how to apply the algorithm described above. 
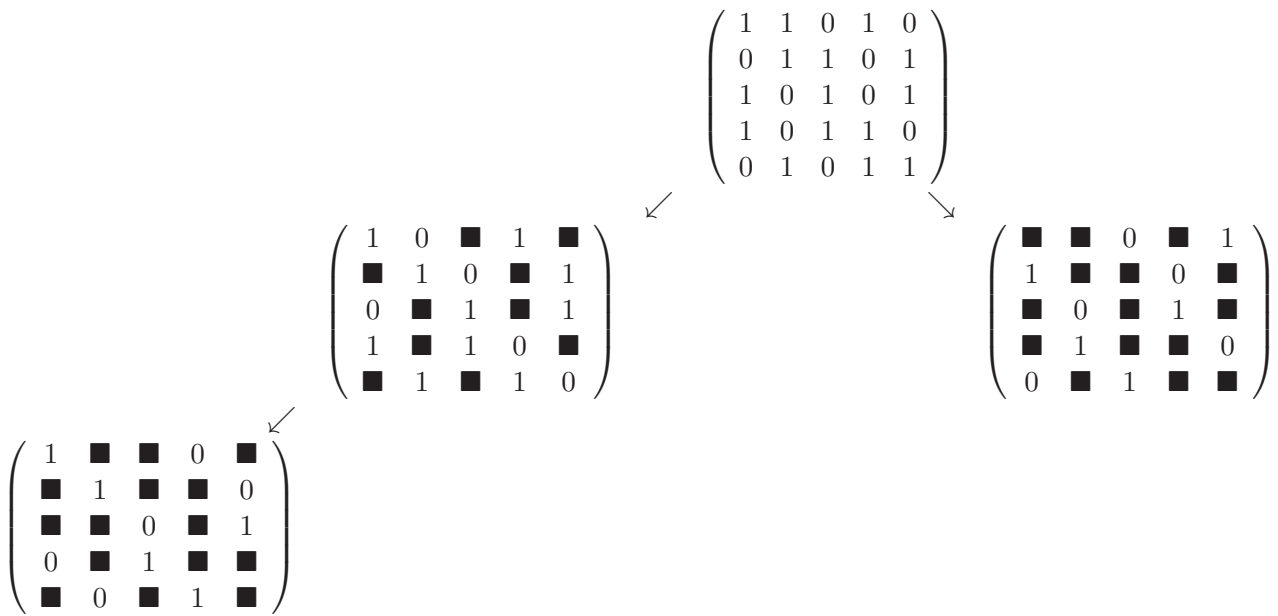
\[ 
\begin{array}{rccrcc}
& & &
\left(
\begin{array}{ccccc}
 1 & 1 & 0 & 1 & 0 \\
 0 & 1 & 1 & 0 & 1 \\
 1 & 0 & 1 & 0 & 1 \\
 1 & 0 & 1 & 1 & 0 \\
 0 & 1 & 0 & 1 & 1
\end{array}
\right)  & & \\
& & \swarrow &  \searrow & & \\
& \left(\begin{array}{ccccc}
 1 & 0 & \blacksquare & 1 & \blacksquare \\
 \blacksquare & 1 & 0 & \blacksquare & 1 \\
 0 & \blacksquare & 1 & \blacksquare & 1 \\
 1 & \blacksquare & 1 & 0 & \blacksquare \\
 \blacksquare & 1 & \blacksquare & 1 & 0
\end{array}\right)
 & & & 
\left( \begin{array}{ccccc}
 \blacksquare & \blacksquare & 0 & \blacksquare & 1 \\
 1 & \blacksquare & \blacksquare & 0 & \blacksquare \\
 \blacksquare & 0 & \blacksquare & 1 & \blacksquare \\
 \blacksquare & 1 & \blacksquare & \blacksquare & 0 \\
 0 & \blacksquare & 1 & \blacksquare & \blacksquare
\end{array} \right)\\
\swarrow &  & & &  \\
\left( \begin{array}{ccccc}
 1 & \blacksquare & \blacksquare & 0 & \blacksquare \\
 \blacksquare & 1 & \blacksquare & \blacksquare & 0 \\
 \blacksquare & \blacksquare & 0 & \blacksquare & 1 \\
 0 & \blacksquare & 1 & \blacksquare & \blacksquare \\
 \blacksquare & 0 & \blacksquare & 1 & \blacksquare
\end{array}\right) \\
\end{array}
\]
This corresponds to the following order $5$ Latin square
\[ \left(\begin{array}{ccccc}
 5 & 1 & 2 & 3 & 4 \\
 4 & 5 & 1 & 2 & 3 \\
 1 & 2 & 3 & 4 & 5 \\
 3 & 4 & 5 & 1 & 2 \\
 2 & 3 & 4 & 5 & 1
\end{array}\right).
\]
\end{example}


\ignore{
\begin{remark}\label{bit_independence}
It is not a priori obvious whether the bits at different levels are independent, and whether certain configurations of binary tables at a given level can potentially be completed by a larger number of Latin squares than other configurations. 
For $n=6$, we are able to provide a negative answer by a simple observation. 
According to the OEIS sequence~A058527~\cite{OEIS}, the number of $6 \times 6$ binary contingency tables with row sums and column sums equal to 3 is 297200, which does not divide the number of Latin squares of order $6$, which is 812851200, see for example~\cite{McKayWanless}.  Thus, some of the  configurations of binary tables must necessarily yield a different number of completable Latin squares. 
Nevertheless, our approach exploits the principles of Boltzmann sampling, which has been previously utilized in various contexts involving contingency tables, see for example~\cite{barvinok2012asymptotic, barvinok2010approximation}. 
\end{remark}
}

\begin{remark}
There is an extensive history related to \emph{partially completed Latin squares of order~$n$}, which are Latin squares of order~$n$ for which only some of the $n^2$ entries have been filled in, and none of those entries a priori violate any of the Latin square conditions. 
In particular, the treatment in~\cite[Chapter~17]{vanLintWilson} related to partially completed Latin squares pertains solely to leaving entire entries filled or unfilled, whereas the algorithm described above is of a distinctly different flavor. 
For example,~\cite[Theorem~17.1]{vanLintWilson} states that all Latin squares of order~$n$ with their first $k$ rows filled in can be extended to a partially completed Latin square with the first $k+1$ rows filled in, for $k=1,2,\ldots,n-1$. 
There are also ways of filling in elements for which no possible Latin square can be realized given those elements; in general, deciding whether a Latin square of order~$n$ with an arbitrary set of squares filled in can be completed is an NP-complete problem, see~\cite{colbourn1984complexity}. 
It would be interesting to explore the analogous analysis for our proposed bit-by-bit approach. 
\end{remark}

\section{Probabilistic approach}
\label{section_probabilistic}
\subsection{Probabilistic divide-and-conquer}

Probabilistic divide-and-conquer (PDC) is an approach for exact sampling by dividing up the sample space into two pieces, sampling each one separately, and then piecing them together appropriately~\cite{PDC}. 
In order to sample using PDC, one starts with a sample space $\Omega$ decomposed into two separate sets $\Omega = \mathcal{A} \times \mathcal{B}$, and writes the target distribution as 
\begin{equation}\label{LS}
\L(S) = \L((A, B) | E),
\end{equation}
 where $A \in \mathcal{A}$ and $B \in \mathcal{B}$ have given distributions and are independent, and $E \subset \mathcal{A}\times \mathcal{B}$ is some measurable event of the sample space.  
Whereas rejection sampling can be described as sampling $a$ from $\L(A)$ and $b$ from $\L(B)$, and then checking whether $(a,b) \in E$ (see Algorithm~\ref{WTGL procedure} below), PDC samples $x$ from $\L(A | E)$ and $y$ from $\L(B | E, a)$, and returns $(x,y)$ as an exact sample (see Algorithm~\ref{PDC procedure} below).

\begin{algorithm}{\rm
\begin{algorithmic}
\State 1.  Generate sample from $\L(A)$, call it $a$.
\State 2.  Generate sample from $\L(B)$, call it $b$.
\State 3.  Check if $(a,b) \in E$; if so, return $(a,b)$, otherwise restart.
\end{algorithmic}}
\caption{\cite{PDC} Standard rejection sampling from $\L((A,B) | E)$}
\label{WTGL procedure}
\end{algorithm}

\begin{algorithm}{\rm
\begin{algorithmic}
\State 1. Generate sample from $\L(A\, |\, E)$, call it $x$.
\State 2. Generate sample from $\L(B\, |\, E, A=x)$ call it $y$.
\State 3. Return $(x,y)$.
\end{algorithmic}}
\caption{\cite{PDC} Probabilistic Divide-and-Conquer sampling from $\L((A,B) | E)$} 
\label{PDC procedure}
\end{algorithm}

\begin{lemma}{PDC Lemma~\cite[Lemma~2.1]{PDC}.}
Suppose $E$ is a measurable event of positive probability. 
Suppose $X$ is a random element of $\A$ with distribution 
\begin{equation}\label{def X}
   \L(X) = \L( \, A \, | \, E \, ),
\end{equation}
and  $Y$ is a random element of $\B$ with conditional distribution
\begin{equation}\label{def Y}
   \L(Y \, | X=a ) = \L( \, B \, | \, E, A=a \, ).
\end{equation}
Then $(X,Y) =^d S$, i.e., the pair $(X,Y)$ has the same distribution  as $S$, given by \eqref{LS}.
\end{lemma}

Often, however, the conditional distributions in Algorithm~\ref{PDC procedure} are not practical to sample from, and so the following PDC algorithm which uses rejection sampling has been proven to be effective and practical in many situations, see for example~\cite{PDCDSH, DeSalvoImprovements, DeSalvoSudoku}.

\begin{algorithm}[H]{\rm
\begin{algorithmic}
\State 1. Generate sample from $\L(A),$ call it $a$.
\State 2. Reject $a$ with probability $1-s(a)$, where $s(a)$ is a function of $\L(B)$, $E$;  otherwise, restart. 
\State 3. Generate sample from $\L(B\, |\, (a,B) \in E),$ call it $y$.
\State 4. Return $(a,y)$.
\end{algorithmic}}
\caption{PDC sampling from $\L( (A,B)\, |\, (A,B)\in E)$ using soft rejection sampling}
\label{PDC procedure von Neumann}
\end{algorithm}

\ignore{
Of course, while we have started with the problem of sampling from one conditional distribution, now we have two!  
Fortunately, by splitting up the randomness into two pieces, which we have the freedom to choose based on available data, it is possible to more effectively target the smaller conditional distribution. 
A large class of PDC algorithms, for example, chooses $\A$ and $\B$ so that, conditional on observing $A = a$, $\{b \in \B : (a,b) \in E\}$ has cardinality at most 1, i.e., the second stage of Algorithm~\ref{PDC procedure} is deterministic, and all of the randomness is contained in the first stage.  
Further details and examples can be found in~\cite{PDCDSH}. 
The example below shows that this idea, decomposing the sample space into two pieces, one of which is deterministic, is a particularly good strategy which requires very little structural knowledge of the sample space, and provides a guaranteed speedup over traditional Boltzmann samplers. 
}

\ignore{
\begin{example}
Let us pause at this point to note a particularly elegant divide-and-conquer strategy for the random sampling of combinatorial structures of the form 
\[\L\left( (X_1, \ldots, X_n) \middle| \sum_{i=1}^n i\, X_i = n\right), \]
where $X_1, X_2, \ldots, X_n$ are independent random variables, typically coming from some family of distributions like Poisson, geometric, or Bernoulli.  
Such structures are typical in combinatorics, see for example~\cite{IPARCS, Boltzmann, Flajolet}, and represent an initial benchmark for improvements to the random sampling of classical combinatorial structures like integer partitions, set partitions, etc.

Consider two divisions: the first is
\[ A_1 = (X_1, X_2, \ldots, X_n), \qquad B_1 = \emptyset \]
and the second is, for any $I \in \{1,2,\ldots,n\}$, 
\[ A_2 = (X_1, \ldots, X_{I-1}, X_{I+1}, \ldots, X_n), \qquad B_2 = (X_I), \]
with $E = \{ \sum_{i=1}^n i\, X_i = n\}$. 
The first division is in fact rejection sampling~\cite{Rejection}, whereas the second has been dubbed PDC deterministic second half, see~\cite{PDCDSH}. 

In order to sample from $\L(A_2 | E) = \L\left((X_1, \ldots,X_{I-1},X_{I+1},\ldots, X_n) \middle| \sum_{i=1}^n i\, X_i = n\right)$, we use what we call \emph{soft rejection sampling}, see~\cite{Rejection}, see also~\cite[Chapter 2]{devroye}. 
That is, we sample from $\L(A_2)$, i.e., the unconditional distribution, and apply a rejection probability to obtain the appropriate conditional distribution. 
The rejection probability is then \emph{proportional} to
\begin{align*}
 \frac{\P( (X_2, \ldots, X_n) = (x_2,\ldots,x_n) | \sum_{i=1}^n i\, X_i = n)}{\P( (X_2, \ldots, X_n) = (x_2,\ldots,x_n) )}& = \frac{\P(X_1 = n - \sum_{i=2}^n i\, X_i | X_2=x_2, \ldots, X_n=x_n)}{\P(\sum_{i=1}^n i\, X_i = n)}. \\
 & = \frac{\P(X_1 = n - \sum_{i=2}^n i\, x_i)}{\P(\sum_{i=1}^n i\, X_i = n)}.
 \end{align*}
Specifically, we obtain the unbiased sampling algorithm in Algorithm~\ref{PDC discrete} below, see~\cite{PDCDSH}, where $U$ denotes a uniform random variable between $0$ and $1$, independent of all other random variables.
\begin{algorithm}
\caption{\cite{PDC, PDCDSH} PDC deterministic second half for sampling from $\L(\, (X_1, X_2, \ldots, X_{n})\ |\ E)$}
\begin{algorithmic}
\State { {\bf input:} Discrete distributions $\L(X_1), \L(X_2), \ldots, \L(X_n), E, I$}
\State Sample from $\L(X_1, \ldots, X_{I-1},X_{I+1},\ldots,X_n)$, denote the observation by $x^{(I)}$.
\State Let $x_I$ denote the unique value for which $(x^{(I)}, x_I) \in E$.
\If {$U< \frac{P\left(X_I = x_I\right)}{\max_\ell P(X_I = \ell)}$ }
		\State {\bf return} $(x^{(I)},x_I)$
\Else
\State {\bf restart}
\EndIf
\end{algorithmic} \label{PDC discrete}
\end{algorithm}

\begin{proposition}
Choose any $I \in \{1, 2, \ldots, n\}$.  Algorithm~\ref{PDC discrete} is faster, in terms of the expected number of random bits required, than an exact Boltzmann sampler by a factor of $(\max_\ell \P(X_I = \ell))^{-1}$.
\end{proposition}
\begin{proof}
The acceptance set for Boltzmann sampling can be formulated as 
\[ \{(X_1, \ldots, X_{I-1}, X_{I+1},\ldots, X_n, U) : \sum_{i\neq I} i\, X_i \leq n \mbox{ and } U < \P(X_I = n - \sum_{i \neq I} i\, X_i) \}, \]
and the acceptance set for PDC deterministic second half can be formulated as 
\[ \left\{(X_1, \ldots, X_{I-1}, X_{I+1},\ldots, X_n, U) : \sum_{i\neq I} i\, X_i \leq n \mbox{ and } U < \frac{\P(X_I = n - \sum_{i \neq I} i\, X_i)}{\max_\ell \P(X_I = \ell)} \right\}.  \qedhere \]
\end{proof}
\end{example}

\begin{remark}
The PDC deterministic second half approach is \emph{not} simply an application of soft rejection sampling, see~\cite[Section~5]{PDCDSH}. 
Soft rejection sampling is indeed a great complementary technique for which sampling from conditional distributions can be made more accessible, but it is by no means necessary in order to apply PDC in general, especially if one can sample from the conditional distributions directly.
\end{remark}
}

\begin{remark}
A surprisingly efficient division occurs when the sets $A$, $B$, and $E$ are such that $ \L( \, B \, | \, E, A=a \, )$ is trivial for each $a \in \mathcal{A}$. 
In this case, it was shown in~\cite{PDC, PDCDSH}, that a nontrivial speedup can be obtained by relatively simple arguments. 
\end{remark}

\begin{remark}
A more recent application in~\cite{DeSalvoImprovements} uses the recursive method~\cite{NW, NWbook} in order to obtain both the rejection probability $s(a)$ along with the conditional distributions $\L(B | (a,B) \in E)$. 
It was proved to have a smaller rejection rate than exact Boltzmann sampling, and also requires less computational resources than the full recursive method. 
\end{remark}

\ignore{
We now recall a division from~\cite{PDC} which is recursive and self-similar, and serves as the inspiration for the algorithm championed in this paper. 
Let $0<x<1$, and let $Z_1(x), Z_2(x), \ldots, Z_n(x)$ denote independent geometric random variables with $Z_i \equiv Z_i(x)$ having distribution which is geometric with parameter $1-x^i$, $i=1,2,\ldots,n$. 
Conditional on $\sum_{i=1}^n i\, Z_i(x) = n$, the random variables $Z_i(x)$, $i=1,2,\ldots, n$, represent the number of parts of size~$i$ in a uniformly chosen integer partition of size~$n$; see~\cite{Fristedt, VershikKerov, Vershik}. 
Our problem is to sample from $\L\left(Z_1(x), \ldots, Z_n(x) \middle | \sum_{i=1}^n i\, Z_i(x) = n\right)$. 

The division championed in this case is inherently probabilistic, and relies on a property of geometric random variables. 
Namely, letting $G(q)$ and $G'(q^2)$ denote independent geometric random variables with parameters $1-q$ and $1-q^2$ respectively, and letting $B(p)$ denote an independent Bernoulli random variable with parameter $p$, we have the decomposition 
\begin{equation}\label{decomposition} G(q) = B\left(\frac{q}{1+q}\right) + 2G'(q^2), 
\end{equation}
where the equality is in distribution. 
In other words, this decomposition shows that the bits in a geometric random variable are independent, and specifies their distribution.
It also suggests a PDC algorithm; that is, sample the least significant bits first, one at a time, and then the remaining part is twice a copy of the first part with a different parameter. 
As a bonus, many sampling algorithms have a tilting parameter which can be chosen arbitrarily, as is the case for parameter $x$ for integer partitions; after each stage of the PDC algorithm, we may adjust the tilting parameter to more optimally target the remaining sampling problem. 

Explicitly, 
we use the division
\[ A = (\epsilon_1, \ldots, \epsilon_n), \qquad B = (2Z_1(x^2), 2Z_2(x^4), \ldots, 2 Z_n(x^{2n})), \]
where $\epsilon_i$ denotes the least significant bit of $Z_i$, distributed as a Bernoulli random variable with parameter $x^i/(1+x^i)$. 
Even though we have written the geometric random variables in $B$ in terms of a new parameter $x^2$, since the distribution we are sampling from is conditionally independent of $x$, we are free to choose any $x' \in (0,1)$ which more optimally targets the sampling distribution, which is realized after the bits in $A$ are sampled. 

An application of Algorithm~\ref{PDC procedure von Neumann} in this setting is the following: Sample from $A = (\epsilon_1, \ldots, \epsilon_n)$ and let $a = \sum_{i=1}^n i\, \epsilon_i$.  
Then we have 
\begin{equation}\label{sa} s(a) = \frac{\P(\sum_{i=1}^n (2i) Z_{2i} = n-a)}{\max_\ell \P( \sum_{i=1}^n (2i)\, Z_{2i} = \ell)} =  \frac{\P\left(\sum_{i=1}^{(n-a)/2} i Z_{i} = \frac{n-a}{2}\right)}{\max_\ell \P\left( \sum_{i=1}^{\ell} i\, Z_{i} = \ell\right)} = \frac{p\left(\frac{n-a}{2}\right)x^{(n-a)/2} \prod_{i=1}^{(n-a)/2} (1-x^i)}{\max_\ell p(\ell) x^\ell \prod_{i=1}^\ell (1-x^i)}, \end{equation}
where $p(n)$ denotes the number of integer partitions of $n$, which is an extremely well-studied sequence; see for example~\cite{PDC} and the references therein. 
In addition, we have $\L(B | (a,B) \in E)$ is equivalent to the original problem with $n$ replaced by $(n-a)/2$, hence the description of this algorithm as self-similar. 
In terms of random sampling of integer partitions, this is also within a constant factor of the fastest possible algorithm; see~\cite[Theorem 3.5]{PDC}.
More importantly, however, is that this type of divide-and-conquer is inherently probabilistic, table-free, and can be applied to other structures which are modeled by geometric random variables, namely,  contingency tables. 

}

Previous applications of PDC have considered mutually independent random variables $X_1, X_2, \ldots, X_n$, with target space given by some measurable event $E$, with the goal to sample from the distribution
\[ \L(\, (X_1, \ldots, X_n)\, | \, E). \]
We assume that the unconditional distributions $\L(X_1), \ldots, \L(X_n)$ are known and can be sampled. 
For general two-dimensional tables, we consider sampling problems of the form: $X_{1,1}, \ldots, X_{m,n}$ are mutually independent random variables, $E$ is some measurable event, and the goal is to the sample from the distribution
\begin{equation}\label{distribution}
 \L\left(\begin{array}{ccc}
X_{1,1}, & \ldots & X_{1,n}, \\
\vdots & \vdots & \vdots \\
X_{m,1}, & \ldots & X_{m,n}
\end{array}\, \middle|\, E \right).
\end{equation}

The uniform distribution over all $(r,c)$-contingency tables can be obtained, see for example~\cite{BarvinokRandomCT, DeSalvoCT}, by taking $X_{i,j}$ to be geometrically distributed with parameter $p_{i,j} = \frac{m}{m+c_j}$, $i=1,\ldots,m$, $j=1,\ldots,n$, subject to the measurable event
\begin{equation}\label{ctE} E \equiv E_{r,c} = \left\{  \ \begin{array}{l} \sum_{\ell=1}^n X_{i,\ell} = r_i, ~\forall ~ 1\leq i\leq m; \\ \\  \sum_{\ell=1}^m X_{\ell,j} = c_j, ~\forall ~1 \leq j \leq n \end{array}\right\}. \end{equation}
Taking $X_{i,j}$ to be Bernoulli distributed with parameter $p_{i,j} / (1+p_{i,j})$, we obtain the uniform distribution over the set of $(r,c)$-binary contingency tables.  

In the case of Latin squares of order~$n$, we consider the parameterization which consists of $n \times n$ tables of values, each row of which is a permutation of elements $\{1,2,\ldots,n\}$, and each column is a permutation of elements $\{1,2,\ldots,n\}$. 
Let $U_{i,j}$ denote a collection of i.i.d.~discrete uniform random variables from the set $\{1,2,\ldots,n\}$, $ 1\leq i, j \leq n$, and let $S_n$ denote the set of all permutations of the elements $\{1,\ldots,n\}$. 
The sample space for random sampling of Latin squares is then given by Equation~\eqref{distribution} with $\L(X_{i,j}) = \L(U_{i,j})$ and measurable event $E$ given by 
\[ E = \left\{  \ \begin{array}{ll} (X_{i,1},\ldots,X_{i,n}) \in S_n, & \forall ~ 1\leq i\leq m; \\ \\  (X_{1,j},\ldots,X_{m,j}) \in S_n, & \forall ~1 \leq j \leq n \end{array}\right\}. \]

Once this probabilistic framework has been established, the next task is to devise a PDC division that is both practical and amenable to the resulting analysis. 
We motivate our PDC division in the next section.

\subsection{An heuristic for contingency tables}

We first recall a division from~\cite{PDC} which is recursive and self-similar, and serves as the inspiration for the algorithm championed in this paper. 
Let $0<x<1$, and let $Z_1(x), Z_2(x), \ldots, Z_n(x)$ denote independent geometric random variables with $Z_i \equiv Z_i(x)$ having distribution which is geometric with parameter $1-x^i$, $i=1,2,\ldots,n$. 
Conditional on $\sum_{i=1}^n i\, Z_i(x) = n$, the random variables $Z_i(x)$, $i=1,2,\ldots, n$, represent the number of parts of size~$i$ in a uniformly chosen integer partition of size~$n$; see~\cite{Fristedt, VershikKerov, Vershik}. 
Suppose, for the purpose of this demonstration, that we wish to sample from $\L\left( (Z_1(x), \ldots, Z_n(x)) \middle | \sum_{i=1}^n i\, Z_i(x) = n\right)$. 

The division championed in this case is inherently probabilistic, and relies on a property of geometric random variables. 
Namely, letting $G(q)$ and $G'(q^2)$ denote independent geometric random variables with parameters $1-q$ and $1-q^2$ respectively, and letting $B(p)$ denote an independent Bernoulli random variable with parameter $p$, we have the decomposition 
\begin{equation}\label{decomposition} G(q) \stackrel{D}{=} B\left(\frac{q}{1+q}\right) + 2G'(q^2), 
\end{equation}
where the equality is in distribution. 
In other words, this decomposition shows that the bits in a geometric random variable are independent, and specifies their distribution.
It also suggests a PDC algorithm; that is, sample the least significant bits first, one at a time, and then the remaining part has a distribution which is two times a geometric distribution with parameter $q^2$ in place of $q$. 
As a bonus, many sampling algorithms have a tilting parameter which can be chosen arbitrarily, as is the case for parameter $x$ for integer partitions; after each stage of the PDC algorithm, we may adjust the tilting parameter to more optimally target the remaining sampling problem. 

Explicitly, 
we use the division
\[ A = (\epsilon_1, \ldots, \epsilon_n), \qquad B = (2Z_1(x^2), 2Z_2(x^4), \ldots, 2 Z_n(x^{2n})), \]
where $\epsilon_i$ denotes the least significant bit of $Z_i$, distributed as a Bernoulli random variable with parameter $x^i/(1+x^i)$. 
Even though we have written the geometric random variables in $B$ in terms of a new parameter $x^2$, since the distribution we are sampling from is conditionally independent of $x$, we are free to choose any $x' \in (0,1)$ which more optimally targets the sampling distribution, which is realized after the bits in $A$ are sampled. 

An application of Algorithm~\ref{PDC procedure von Neumann} in this setting is the following: Sample from $A = (\epsilon_1, \ldots, \epsilon_n)$ and let $a = \sum_{i=1}^n i\, \epsilon_i$.  
Then we have 
\begin{equation}\label{sa} s(a) = \frac{\P(\sum_{i=1}^n (2i) Z_{2i} = n-a)}{\max_\ell \P( \sum_{i=1}^n (2i)\, Z_{2i} = \ell)} =  \frac{\P\left(\sum_{i=1}^{(n-a)/2} i Z_{i} = \frac{n-a}{2}\right)}{\max_\ell \P\left( \sum_{i=1}^{\ell} i\, Z_{i} = \ell\right)} = \frac{p\left(\frac{n-a}{2}\right)x^{(n-a)/2} \prod_{i=1}^{(n-a)/2} (1-x^i)}{\max_\ell p(\ell) x^\ell \prod_{i=1}^\ell (1-x^i)}, \end{equation}
where $p(n)$ denotes the number of integer partitions of $n$, which is an extremely well-studied sequence; see for example~\cite{PDC} and the references therein. 
In addition, we have $\L(B | (a,B) \in E)$ is equivalent to the original problem with $n$ replaced by $(n-a)/2$, hence the description of this algorithm as self-similar. 
In terms of random sampling of integer partitions, this is also within a constant factor of the fastest possible algorithm due to the efficient computation of $p(n)$; see~\cite[Theorem 3.5]{PDC}.
More importantly, however, is that this type of divide-and-conquer is inherently probabilistic, table-free, and can be applied to other structures which are modeled by geometric random variables, namely,  contingency tables. 

Next, we revisit and generalize an approach in~\cite{DeSalvoCT} for random sampling of nonnegative integer-valued $(r,c)$-contingency tables inspired by the above application for integer partitions. 
The generalization fixes an arbitrary set of entries to be 0, which doesn't change the approximation heuristic; however, it does a priori drastically change the computational complexity of computing the exact sampling rejection probability. 

We start with a well-known probabilistic model for the entries in a random contingency table, generalized to include entries which are forced to be 0.
In this section, we let $W$ be any given $m\times n$ matrix with values in $\{0,1\}$. 
For each $j=1,2,\ldots,n$, we define $J_j := \{1,\ldots, n\} \setminus \{ i : W(i,j) = 1\}$. 
Similarly, for each $i=1,2,\ldots,m$, we define $I_i := \{1,\ldots, n\} \setminus \{j : W(i,j) = 1\}$. 
Also, we let $h_j = \sum_{i=1}^m W(i,j)$ denote the number of entries forced to be zero in column $j$, for $j=1,2,\ldots,n$.

\begin{lemma}
\label{lemma:uniform}
Let $\X_W=(X_{ij})_{1 \leq i \leq m, 1 \leq j \leq n}$ denote a collection of independent geometric random variables with parameters $p_{ij}$, such that $W(i,j) = 1$ implies $X_{i,j}$ is conditioned to have value 0.   
If $p_{ij}$ has the form $p_{ij} = 1 - \alpha_i \beta_j$, then $\X_W$ is uniform restricted to $(r,c)$-contingency tables with zeros in entries indicated by $W$.
\end{lemma}
\begin{proof}
Let $\X=(X_{ij})_{1 \leq i \leq m, 1 \leq j \leq n}$ denote a collection of independent geometric random variables with parameters $p_{ij}$, where $p_{ij}$ has the form $p_{ij} = 1 - \alpha_i \beta_j$. 
We have 
\[\P\big(\X=\xi\big)
=\prod_{i,j}\P\big(X_{ij}=\xi_{ij}\big)
=\prod_{i,j}(\alpha_i \beta_j)^{\xi_{ij}}(1-\alpha_i \beta_j)
=\prod_i\alpha_i^{r_i} \prod_j\beta_j^{c_j}\prod_{i,j}\big(1-\alpha_i \beta_j\big).
\]
Since this probability does not depend on $\xi$, it follows that the restriction of $\X$ to $(r,c)$-contingency tables is uniform. 
As the collection of random variables are independent, conditioning on any $X_{i,j} = 0$ only changes the constant of proportionality, and does not affect the dependence on the $\xi$, hence 
\[\P\big(\X_W=\xi\big)
=\prod_{i,j : W(i,j) = 0}\P\big(X_{ij}=\xi_{ij}\big)
=\prod_{i,j: W(i,j)=0}(\alpha_i \beta_j)^{\xi_{ij}}(1-\alpha_i \beta_j)
=\prod_i\alpha_i^{r_i} \prod_j\beta_j^{c_j}\prod_{i,j}\big(1-\alpha_i \beta_j\big);
\]
i.e., it follows that the restriction of $\X_W$ to $(r,c)$-contingency tables with forced zero entries indicated by $W$ is uniform. 
\end{proof}

\begin{lemma}\label{expectation}
Suppose $\X_W$ is a collection of independent geometric random variables, where $X_{i,j}$ has parameter $p_{ij} = \frac{m-h_j}{m-h_j+c_j}$, for all pairs $(i,j)$ such that $W(i,j) = 0$, and $p_{i,j} = 1$ for all pairs $(i,j)$ such that $W(i,j) = 1$.  Then the expected column sums of $\X_W$ are $c_1, c_2, \ldots, c_n$, and the expected row sums are $\sum_{j\in I_1} \frac{c_j}{m-h_j}, \ldots, \sum_{j\in I_m} \frac{c_j}{m-h_j}$. 
\end{lemma}
\begin{proof}
Note first that $\E X_{i,j} = p_{i,j}^{-1} - 1$.  Then 
for any $j=1,2,\ldots,n$,
\[ \sum_{i=1}^m \E\, X_{i,j} = \sum_{i \in J_j} \frac{m-h_j+c_j}{m-h_j} - 1= c_j  \]
and similarly for any $i=1,2,\ldots, m$, 
\[ \sum_{j=1}^n \E\, X_{i,j} = \sum_{j \in I_i} \frac{m-h_j+c_j}{m-h_j}-1 = \sum_{j\in I_i} \frac{c_j}{m-h_j}. \qedhere\]
\end{proof}

Next, we describe the PDC algorithm and how the heuristic is utilized.

Suppose we have matrices $W$ and $\mathcal{O}$ with entries in $\{0,1\}$.  
For any set of row sums and column sums $(r,c)$, let $\Sigma(r,c,\mathcal{O}, W)$ denote the number of nonnegative integer-valued $(r,c)$-contingency tables with entry $(i,j)$ forced to be even if the $(i,j)$th entry of $\mathcal{O}$ is 1, and entry $(i,j)$ forced to be 0 if the $(i,j)$th entry of $W$ is 1. 
Let $\mathcal{O}_{i,j}$ denote the matrix which has entries with value $1$ in the first $j-1$ columns, and entries with value $1$ in the first $i$ rows of column $j$, and entries with value 0 otherwise. 

The algorithm we champion for contingency tables starts by considering the first entry in the first column not forced to be 0, say at entry $(s,1)$, and denote the least significant bit by $\epsilon_{s,1}$. 
As we are in the first column, we have $\L(\epsilon_{s,1} | E) = \L\left(\Bern\left(\frac{q_1}{1+q_1}\right) \middle| E\right)$, where $E$ is given in Equation~\eqref{ctE}. 
In lieu of sampling from this conditional distribution directly, and a more practical alternative is to sample from $\L\left(\Bern\left(\frac{q_1}{1+q_1}\right)\right)$ and reject the outcome with probability proportional to $1-\P(E | \epsilon_{s,1})$. 

We define $r_i(k) = (r_1, \ldots, r_i - k, \ldots r_m)$ and $c_j(k) = (c_1, \ldots, c_j - k, \ldots c_n)$ for $k \in \{0,1\}$. 
By Lemma~\ref{lemma:uniform}, we have 
\[ \P(E | \epsilon_{s,1} = k) = \Sigma(r_s(k),c_1(k),W,\mathcal{O}_{i,j})\,\cdot \, \left(q_1^{c_1-k}\, (1-q_1^2)\right) (1-q_1)^{m-h_1} \prod_{j=2}^n q_j^{c_j} (1-q_j)^{m-h_j}. \]
Then, since we reject \emph{in proportion} to this probability, we normalize by all terms which do not depend on $k$, which gives
\begin{equation}\label{enumerative}
 \P(E | \epsilon_{s,1} = k) \propto \Sigma(r_s(k),c_1(k),W,\mathcal{O}_{i,j})\,\cdot \, q_1^{-k},
\end{equation}
where we use the notation $a(k) \propto b(k)$ to mean that $a(k)/b(k)$ is equal to a positive constant independent of $k$ (but potentially depending on all other parameters). 
Once we accept a value for $\epsilon_{s,1}$, we then move to the next bit further down in the column and sample analogously, continuing in this manner column by column, left to right. 

The following is a probabilistic formulation of the \emph{exact} rejection probability for an arbitrary entry $(i,j)$, which assumes that all bits above the entry in the current column have already been sampled, as well as all bits in all columns to the left of $(i,j)$; see Section~\ref{section_probabilistic} for the complete description of each random variable:
 \begin{align}\label{fij}
\hspace{-0.25in} f(i,j,k, r, c,W)  \propto \ &
 \P\left( \begin{array}{llll}  
\sum_{\ell\in J_1}^{\ell < j} 2\xi''_{1,\ell}(q_\ell^2, c_\ell) &+ \eta_{1,j,i}'(q_j, c_j) &+  \sum_{\ell\in J_1}^{\ell > j} \xi'_{1,j}(q_\ell, c_\ell) &= r_1 \\
\sum_{\ell\in J_2}^{\ell < j} 2\xi''_{2,\ell}(q_\ell^2, c_\ell) &+ \eta_{2,j,i}'(q_j, c_j) &+  \sum_{\ell\in J_2}^{\ell > j} \xi'_{2,j}(q_\ell, c_\ell) &= r_2 \\
\qquad \vdots \\
\sum_{\ell\in J_{i-1}}^{\ell < j} 2\xi''_{i-1,\ell}(q_\ell^2, c_\ell) &+ \eta_{i-1,j,i}'(q_j, c_j) &+  \sum_{\ell\in J_{i-1}}^{\ell > j} \xi'_{i-1,j}(q_\ell, c_\ell) &= r_{i-1} \\
\sum_{\ell\in J_i}^{\ell < j}2\xi''_{i,\ell}(q_\ell^2, c_\ell) &+ \eta_{i,j,i}''(q_j, c_j) &+  \sum_{\ell\in J_i}^{\ell > j} \xi'_{i,j}(q_\ell, c_\ell) &= r_{i}-k \\
\sum_{\ell\in J_{i+1}}^{\ell < j}2\xi''_{i+1,\ell}(q_\ell^2, c_\ell) &+ \eta_{i+1,j,i}''(q_j, c_j) &+  \sum_{\ell\in J_{i+1}}^{\ell > j} \xi'_{i+1,j}(q_\ell, c_\ell) &= r_{i+1} \\
\qquad \vdots \\
\sum_{\ell\in J_m}^{\ell < j}2\xi''_{m,\ell}(q_\ell^2, c_\ell) &+ \eta_{m,j,i}''(q_j, c_j) &+  \sum_{\ell\in J_m}^{\ell > j} \xi'_{m,j}(q_\ell, c_\ell) &= r_{m} \\
\end{array}\right) \\
\nonumber & \ \ \ \  \times  \P\left( \sum_{\ell\in I_i, \ell \leq i} 2\, \xi_{\ell,j}(q_j^2) + \sum_{\ell\in I_i, \ell > i} \xi_{\ell,j}(q_\ell) = c_j - k \right). 
\end{align} 
If we could evaluate the multivariate probability above exactly, or to some arbitrarily defined precision, then the bit-by-bit sampling approach would yield an exact sampling algorithm. 
However, we champion the following approximation to Equation~\eqref{fij} based on the heuristic that the dependencies between the random variables are strongest along the $i$-th row and the $j$-th column:
\begin{align}
\nonumber & F(i,j,k, r, c,W) := \\
\label{alternative:rejection} &\qquad  \Pr\left( \sum_{\ell \in I_i, \ell < j} 2\xi''_{i,\ell}(q_\ell^2, c_j, \epsilon) + 2\eta''_{i,j,i} + \sum_{\ell \in I_i, \ell > j} \xi'_{i,\ell}(q_\ell) = r_i - k\right) \\
\nonumber &\qquad  \ \ \ \times \Pr\left( \sum_{\ell\in J_j, \ell \leq i} 2\xi_{\ell,j}(q_j^2) + \sum_{\ell \in J_j, \ell > i} \xi_{\ell,j}(q_j) = c_j - k\right), \qquad \qquad k \in \{0,1\}. 
 \end{align}
Note that while this approximation does not encode the essential global information, it nonetheless enforces several necessary conditions like parity of the row and column, and is fast to compute directly. 


\subsection{An heuristic for Latin squares}
\label{latin_square_heuristic}

The following is a variation of the motivating example of the previous section. 
Let $0<x<1$, and let $X_1(x), X_2(x), \ldots, X_n(x)$ denote independent Bernoulli random variables with $X_i \equiv X_i(x)$ having distribution which is Bernoulli with parameter $\frac{x^i}{1+x^i}$, $i=1,2,\ldots,n$. 
Conditional on $\sum_{i=1}^n i\, X_i(x) = n$, the random variables $X_i(x)$, $i=1,2,\ldots, n$, represent the number of parts of size~$i$ in a uniformly chosen integer partition of size~$n$ \emph{into distinct part sizes}; see~\cite{Fristedt, LD}. 
As before, suppose we wish to sample from $\L\left( (X_1(x), \ldots, X_n(x)) \middle | \sum_{i=1}^n i\, X_i(x) = n\right)$. 
We then consider the following ``spatial" PDC division
\[ A = (X_1, X_3, X_5, \ldots), \qquad B = (X_2, X_4, X_6, \ldots). \]
The result is again a self-similar, recursive PDC algorithm, with the rejection function given by 
\begin{equation}\label{sa} s(a) = \frac{\P(\sum_{i=1}^{n/2} (2i) X_{2i} = n-a)}{\max_\ell \P( \sum_{i=1}^{n/2} (2i)\, X_{2i} = \ell)} =  
\frac{p_d\left(\frac{n-a}{2}\right)x^{(n-a)/2} \prod_{i=1}^{(n-a)/2} (1+x^i)^{-1}}{\max_\ell p_d(\ell) x^\ell \prod_{i=1}^\ell (1+x^i)^{-1}}, \end{equation}
where $p_d(n)$ is the number of integer partitions into distinct part sizes, which can be computed efficiently via, e.g.,~\cite{hua1942number}. 

Let $\LS_n$ denote the set of Latin squares of order $n$, and let $\BC_{n,m}$ denote the set of all $n\times n$ binary contingency tables with line sums $m$.
Consider any Latin square of order~$n$, and consider the map $\varphi_n : \LS_n \longrightarrow \BC_{n,\lceil n/2\rceil}$ which replaces each entry having value $x$ with $x \mod 2$, i.e., each entry becomes a $1$ if it is odd and $0$ if it is even. 
The result is an $n\times n$ binary contingency table with row sums and column sums $\lceil n/2 \rceil$.
Let us assume as before that $n$ is a perfect power of 2 for simplicity of exposition. 

By symmetry, each entry in a Latin square is equally likely to be even or odd, subject to the condition that there are an equal number of odd and even entries in each row and column. 
It is thus natural to conjecture that the joint distribution of the parity of entries in a random Latin square of order~$n$ is close in some sense to a joint distribution of $n\times n$ i.i.d.~Bernoulli random variables with parameter $1/2$, subject to an equal number of 0s and 1s in each row and column; i.e., a uniform distribution on the corresponding set of binary contingency tables $\BC_{n,n/2}$. 
In fact, the standard approach to sampling from binary contingency tables, given in~\cite{chen}, does in fact model each entry as independent subject to the constraints, with entry $(i,j)$ as having a Bernoulli distribution with parameter $\frac{c_j}{m}$, where $c_j$ is the column sum and $m$ is the number of rows; taking $c_j = n/2$ and $m=n$, we obtain our aforementioned natural heuristic for the parity of entries in a Latin square of order $n$. 

We may continue this reasoning for tables with prescribed entries forced to be zero, which is the inspiration for Algorithm~\ref{latin:square:algorithm}.  
We similarly have analogous lemmas as in the previous section specifically for binary-valued tables, which are also straightforward generalizations from the work in~\cite{chen}, and so we omit the proof. 
\begin{lemma}
\label{lemma_uniform_binary}
Let $\X_W=(X_{ij})_{1 \leq i \leq m, 1 \leq j \leq n}$ denote a collection of independent Bernoulli random variables with parameters $p_{ij}$, such that $W(i,j) = 1$ implies $X_{i,j}$ is conditioned to have value 0.   
If $p_{ij}$ has the form $p_{ij} = 1 - \alpha_i \beta_j$, then $\X_W$ is uniform restricted to $(r,c)$-binary contingency tables with zeros in entries indicated by $W$.
\end{lemma}

\begin{lemma}\label{expectation_binary}
Suppose $\X_W$ is a collection of independent Bernoulli random variables, where $X_{i,j}$ has parameter $p_{ij} = \frac{c_j}{m-h_j}$, for all pairs $(i,j)$ such that $W(i,j) = 0$, and $p_{i,j} = 1$ for all pairs $(i,j)$ such that $W(i,j) = 1$.  Then the expected column sums of $\X_W$ are $c_1, c_2, \ldots, c_n$, and the expected row sums are $\sum_{j\in I_1} \frac{c_j}{m-h_j}, \ldots, \sum_{j\in I_m} \frac{c_j}{m-h_j}$. 
\end{lemma}

There are thus two distinct places where bias can occur in our approach championed in Section~\ref{novel_algorithm}. 
The first comes from the locations of the even/odd entries; even if we suppose that we may sample from $\BC_{n,n/2}$ uniformly at random, there is no a priori reason to believe that $\varphi_n$ maps the same number of Latin squares of order $n$ to each element of $\BC_{n,n/2}$, and indeed this is easily seen to \emph{not} be the case for $n=5$, even though the distribution is close to uniform. 
However, supposing we did know precisely the multiset of counts for the image of $\varphi_n$, a natural sampling approach would be to generate samples from $\BC_{n,n/2}$ uniformly and apply an appropriate rejection proportional to this multiset of counts (followed by the analogous procedure for the follow-up tables in the algorithm); however, sampling uniformly from $\BC_{n,n/2}$, is not always straightforward or efficient, which motivates the following approximation algorithm. 

In~\cite{chen}, the algorithm championed for binary contingency tables keeps track of an importance sampling weight, which represents the bias introduced by the sampling algorithm.  
The algorithm proceeds in the same manner as the one we champion for contingency tables, sampling entries column by column from top to bottom. 
Our suggested modification for binary contingency tables is to ignore the use of the Poisson-binomial distribution altogether, but still apply a rejection which captures several necessary conditions, along with what we suspect is the dominant source of bias in the sampling algorithm. 
That is, letting $X_{i,j}(p_{i,j})$ denote a Bernoulli random variable with parameter $p_j \equiv p_{i,j} = \frac{c_j}{n}$, we propose a rejection probability of the form 
\begin{align}
\label{bct_rejection}
B(i,j,k, r, c,W) & :=  \Pr\left(\sum_{\ell \in I_i, \ell > j} X_{i,\ell}(p_\ell) = r_i - k\right) \times \Pr\left(\sum_{\ell \in J_j, \ell > i} X_{\ell,j}(p_j) = c_j - k\right), \qquad  k \in \{0,1\}. 
 \end{align}
It is plausible that for the random sampling of Latin squares the two sources of bias at times cancel each other out, or potentially exacerbate the individual biases; it would certainly be interesting to explore this in more detail.  

\begin{remark}
One can also ask the following fundamental question: Is the range of the injection $\varphi_n$ equal to $\BC_{n,\lceil n/2\rceil}$ for all $n$?
In other words, do there exist binary contingency tables which, if generated, doom our approach from the start? 
The unfortunate answer is yes, found by exhaustive search for $n=7$, although in practice we have not found there to be very many. 
The next natural question, of interest in complexity theory, would be if there are necessary and sufficient conditions on $\BC_{n,\lceil n/2\rceil}$ which indicate in advance whether such a configuration of even/odd entries cannot be completed. 
One could also ask similar questions pertaining to the later phases of the algorithm. 
\end{remark}

\section{Algorithms}
\label{algorithms}

\subsection{Nonnegative integer-valued contingency tables with prescribed 0 entries}

\ignore{
Suppose we have matrices $W$ and $\mathcal{O}$ with entries in $\{0,1\}$.  
For any set of row sums and column sums $(r,c)$, let $\Sigma(r,c,\mathcal{O}, W)$ denote the number of nonnegative integer-valued $(r,c)$-contingency tables with entry $(i,j)$ forced to be even if the $(i,j)$th entry of $\mathcal{O}$ is 1, and entry $(i,j)$ forced to be 0 if the $(i,j)$th entry of $W$ is 1. 
Let $\mathcal{O}_{i,j}$ denote the matrix which has entries with value $1$ in the first $j-1$ columns, and entries with value $1$ in the first $i$ rows of column $j$, and entries with value 0 otherwise. 
Let $\Sigma(r, c, \mathcal{O}, W; t_1, \ldots, t_r)$, where each $t_\ell$, $\ell = 1, \ldots, r$ denotes an entry in the table, say $t_\ell = (i_\ell, j_\ell)$, denote the number of nonnegative integer-valued $(r,c)$-contingency tables as before with the additional assumption that the entries $t_1, \ldots t_r$ are forced to be even. 
}
\ignore{
Let $J_j \leftarrow \{1,\ldots, n\} \setminus \{ i : W(i,j) = 1\}$, and let $k_j$ denote the last entry in $J_j$, for all $j=1,2,\ldots,n$. 
Let $I_i \leftarrow \{1,\ldots, n\} \setminus \{j : W(i,j) = 1\}$, and let $\ell_i$ denote the last entry in $I_i$, for all $i=1,2,\ldots,m$. 

For each $1 \leq i \leq m, 1 \leq j \leq n$, define
\begin{equation}\label{rejection:ij}\small
 f(i,j,k,r,c) :=  \frac{\Sigma\left( \begin{array}{l}(\ldots,r_{i}-k,\ldots), \\ (\ldots,c_{j}-k,\ldots), \\ W\end{array} \right)}{\Sigma\left( \begin{array}{l}(\ldots,r_{i}-1,\ldots), \\ (\ldots,c_{j}-1,\ldots), \\ W\end{array} \right) + \Sigma\left( \begin{array}{l}(\ldots,r_{i},\ldots), \\ (\ldots,c_{j},\ldots), \\ W\end{array} \right)}
\end{equation}
if $i \in J_j \setminus \{k_j\}$, and $j \in I_i \setminus \{\ell_i\}$. 

Let $g_i := \sum_{j=1}^n W_{i,j}$ denote the number of non-zero elements row $i$ of $W$, for $i=1,\ldots,m$, and let $h_j := \sum_{i=1}^m W_{i,j}$ denote the number of non-zero elements in column $j$, $j=1,2,\ldots,n$.
Let $q_j := \frac{c_j}{m-h_j+c_j}$, and 
let $y_j := q_j^{-1} = 1+\frac{m-h_j}{c_j}$.  
Also, let $k_j'$ denote the penultimate entry in $J_j$, for all $j=1,2,\ldots,n$, and let $\ell_i'$ denote the penultimate entry in $I_i$, for all $i=1,2,\ldots,m$. 

Let $b(k)$ be such that $r_i-k-b(k)$ is even.  
Then we define for each $j =1,2,\ldots,n$ and $i \in J_j \setminus \{k_j\}$, 
\begin{equation} \label{equation:in} f(i,\ell_i',k,r,c) := \end{equation}
\begin{equation*}\small \frac{\Sigma\left( \begin{array}{l}(\ldots,r_{i}-k-b(k),\ldots), \\ (\ldots,c_{\ell_i'}-k,c_{\ell_i}-b(k)), \\ W\end{array} \right)\,\cdot\, y_{\ell_i}^{b(k)}}{\Sigma\left( \begin{array}{l}(\ldots,r_{i}-1-b(1),\ldots), \\ (\ldots,c_{\ell_i'}-1,c_{\ell_i}-b(1)), \\ W\end{array} \right)\,\cdot\, y_{\ell_i}^{b(1)}+\Sigma\left( \begin{array}{l}(\ldots,r_{i}-b(0),\ldots), \\ (\ldots,c_{\ell_i'},c_{\ell_i}-b(0)), \\ W\end{array} \right)\,\cdot\, y_{\ell_i}^{b(0)}}
\end{equation*}

Let $v(k)$ be such that $c_j - k-v(k)$ is even. 
Similarly we define for each $i=1,2,\ldots,m$ and $j \in I_i \setminus \{\ell_i\}$, 
\begin{equation}\small \label{rejection:mj}
f(k_j',j,k,r,c) := 
\end{equation}
\begin{equation*}
 \frac{\Sigma\left( \begin{array}{l} (\ldots,r_{k_j'}-k,r_{k_j}-v(k)), \\ (\ldots,c_{j}-k-v(k),\ldots), \\ W\end{array}\right)\,\cdot\, y_{j}^{v(k)}}{\Sigma\left( \begin{array}{l} (\ldots,r_{k_j'}-1,r_{k_j}-v(1)), \\ (\ldots,c_{j}-1-v(1),\ldots), \\ W\end{array}\right)\,\cdot\, y_{j}^{v(1)}+\Sigma\left( \begin{array}{l} (\ldots,r_{k_j'},r_{k_j}-v(0)), \\ (\ldots,c_{j}-v(0),\ldots), \\ W\end{array}\right)\,\cdot\, y_{j}^{v(0)}}. 
\end{equation*}

Finally, let $v(k)$ be such that $c_{n-1} - k-v(k)$ is even, let $\gamma(k)$ be such that $r_{m-1}-k-\gamma(k)$ is even, and let $b(k)$ be such that $c_n-\gamma-b(k)$ is even.  Then we define
\begin{equation*}
A \equiv \Sigma\left(\begin{array}{l}(\ldots,r_{m-1}-k-\gamma(k),r_m-v(k)-b(k)), \\ (\ldots,c_{n-1}-k-v(k),c_n-\gamma(k)-b(k)), \\ \mathcal{O}_{m,n}\end{array} \right)\,\cdot \, y_{n-1}^{v(k)}\, y_n^{\gamma(k) + b(k)},
\end{equation*}
\begin{equation*}
B \equiv \Sigma\left(\begin{array}{l}(\ldots,r_{m-1}-1-\gamma(1),r_m-v(1)-b(1)), \\ (\ldots,c_{n-1}-k-v(1),c_n-\gamma(1)-b(1)), \\ \mathcal{O}_{m,n}\end{array} \right)\,\cdot \, y_{n-1}^{v(1)}\, y_n^{\gamma(1) + b(1)},
\end{equation*}
\begin{equation*}
C \equiv \Sigma\left(\begin{array}{l}(\ldots,r_{m-1}-\gamma(0),r_m-v(0)-b(0)), \\ (\ldots,c_{n-1}-v(0),c_n-\gamma(0)-b(0)), \\ \mathcal{O}_{m,n}\end{array} \right)\,\cdot \, y_{n-1}^{v(0)}\, y_n^{\gamma(0) + b(0)},
\end{equation*}
\begin{equation}\label{equation:mn} f(m-1,n-1,k,r,c) :=  \frac{A}{B+C}. \end{equation}

and let $\ldots$.
}

For any $ 1 \leq i \leq m$, $1 \leq j \leq n$, recall the definition of $\Sigma(r, c, \mathcal{O}_{i,j}, W)$ from the previous section. 
For $k=0,1$, let $r_i(k) = (r_1,\ldots,r_i-k,\ldots,r_m)$, and $c_j(k) = (c_1,\ldots, c_j-k, \ldots,c_n)$. 
Then define
\begin{align*}
 A_{i,j}(k) & := \Sigma\left(r_i(k), c_j(k), \mathcal{O}_{i,j}, W\right), 
\end{align*}
\begin{equation}\label{equation:mn} g(i,j,k,r,c,W) :=  \frac{A_{i,j}(k)}{A_{i,j}(0)+A_{i,j}(1)}, \qquad k \in \{0,1\}. \end{equation}

 The following notations and definitions from~\cite{DeSalvoCT} are used in the algorithms. 
 
 \begin{tabular}{p{1.9cm}p{13cm}}
$\Geo(q)$ & Geometric distribution with probability of success $1-q$, for $0<q<1$, with $\Pr\left(\Geo(q)=k\right)=(1-q)q^k$, $k=0,1,2,\ldots\, $. \\[.15cm]
$\NB(m, q)$ & Negative binomial distribution with parameters $m$ and $1-q$, given by the sum of $m$ independent $\Geo(q)$ random variables, with 
\[\Pr(\NB(m,q) = k) = \binom{m+k-1}{k}(1-q)^m q^k.\] \\[.15cm]
U 
& Uniform distribution on $[0,1]$. We will also denote random variables from this distribution as $U$ or $u$; should be considered independent of all other random variables, including other instances of $U$. \\[.15cm]
$\Bern(p)$ & Bernoulli distribution with probability of success $p$. Similarly to $U$, we will also use it as a random variable. \\[.15cm]
$\xi_{i,j}(q)$ & $\Geo(q)$ random variables which are independent for distinct pairs $(i,j)$, $1\le i\le m$, $1\le j\le n$. \\[.15cm]
$\xi_{i,j}'(q,c_j,W)$ & Random variables which have distribution \[\L\left(\xi_{i,j}(q) \bigg| \sum_{\ell\in J_j} \xi_{\ell,j}(q) = c_j\right),\] and are independent of all other random variables $\xi_{i,\ell}(q)$ for $\ell \ne j$. \\[.15cm]
$2\xi_{i,j}''(q,c_j,W)$ & Random variables which have distribution \[\L\left(2\xi_{i,j}(q^2) \bigg| \sum_{\ell\in J_j} 2\xi_{\ell,j}(q^2) = c_j \right)\] 
and are independent of all other random variables $\xi_{i,\ell}(q)$ for $\ell \ne j$. \\[.15cm]
$\eta'_{i,j,s}(q, c_j)$ & Random variables which have distribution 
\[ \L\left( \xi_{i,j}(q) \middle| \sum_{\ell\in J_j, \ell \leq s} 2\xi_{\ell,j}(q^2) + \sum_{\ell\in J_j, \ell > s} \xi_{\ell,j}(q) = c_j  \right), \]
and are independent of all other random variables $\xi_{i,\ell}(q)$ for $\ell \ne j$. \\[.15cm]
\end{tabular}

 \begin{tabular}{p{1.9cm}p{13cm}}
$2\eta''_{i,j,s}(q, c_j)$ & Random variables which have distribution 
\[ \L\left( 2\xi_{i,j}(q^2) \middle| \sum_{\ell\in J_j, \ell \leq s} 2\xi_{\ell,j}(q^2) + \sum_{\ell\in J_j, \ell > s} \xi_{\ell,j}(q) = c_j  \right), \]
and are independent of all other random variables $\xi_{i,\ell}(q)$ for $\ell \ne j$. \\[.15cm]
$\mathbf q$ & The vector $(q_1, \ldots, q_n)$, where $0<q_i<1$ for all $i=1,\ldots,n$. \\[.15cm]
\end{tabular}

Consider next the following subroutines to generate a random bit:
\begin{algorithm}[H]
\label{bit_exact}
\caption*{Exact sampling of least significant bit of entry $(i,j)$ via enumeration of contingency tables}
\begin{algorithmic}[1]
\Procedure{Exact}{$i,j,r,c,A,W_b$}
  	  \If{ $u \leq g(i,j,0,r,c,W_b) $}\label{reject_g}
	       \State {\bf return} $0$
	  \Else
	       \State {\bf return} $1$
	  \EndIf
\EndProcedure 
\end{algorithmic}
\end{algorithm}

\begin{algorithm}[H]
\label{bit_exact_prob}
\caption*{Exact sampling of least significant bit of entry $(i,j)$ via exact rejection sampling}
\begin{algorithmic}[1]
\Procedure{Exact\_Probabilistic}{$i,j,r,c,A,W_b$}
  	  \If{ $u \leq f(i,j,0,r,c,W_b) $}\label{reject_f}
	       \State {\bf return} $0$
	  \Else
	       \State {\bf return} $1$
	  \EndIf
\ignore{\State $X(i,j) \leftarrow \mbox{Bern}(c_j / (m-h_j))$ \label{random_step}
\If{$U >  \frac{f(i,j,X(i,j),r,c,W_b)}{\max_{\ell \in \{0,1\}f(i,j,\ell,r,c,A,W_b)}}$}
	\State Goto Line~\ref{random_step}.
\EndIf
\State {\bf return} $X_{i,j}$
}
\EndProcedure 
\end{algorithmic}
\end{algorithm}

The previous two procedures are computationally intensive to evaluate, since they a priori rely on the entire solution set, i.e., the set of all potential values of entries. 
We suggest the following approximation approach, in which the $(i,j)$-th entry is decided based solely on the row sum $r_i$ and column sum $c_j$. 
For each fixed $1 \leq i \leq m$ and $ 1 \leq j \leq n$, let $t_{1}, \ldots, t_{\nu}$ denote the set of entries and the corresponding value which have a deterministic value (either the entire entry or the least significant bit) determined by the line sum conditions after the sampling of the least significant bit of entry $(i,j)$, and let $r_{i,j}(t_1,\ldots, t_\nu)$ and $c_{i,j}(t_1, \ldots, t_\nu)$ denote the residual row and column sums, respectively, given these entries.  
For example, if there is only one fillable entry in a given row or column, we may uniquely specify its value; 
if a row sum or column sum is zero, set all fillable entries in that row or column to zero; if there are two entries left in a given row and the row sum is even (and non-zero), then either both least significant bits are 0 or both are 1. 
\ignore{
\begin{enumerate}
\item Fill in a set of values $\{t_1, \ldots, t_\nu\}$ into the table.
\item Check if a priori the entries of any rows or columns are uniquely determined given these values.  There are two cases: if a row sum is zero then all remaining entries must be 0, or if a row sum equals the number of undecided entries then they must all be 1.   Similarly with columns.  
\item If any values were filled in from the previous step, call the function recursively; stop when no new elements are filled in, or when $W$ is the matrix of all 1s.
\end{enumerate}}

\begin{algorithm}[H]
\label{bit_approx}
\caption*{Given entries with values to fill in, return the set of determined entries.}
\begin{algorithmic}[1]
\Procedure{Deterministic\_Fill}{$\{t_1, \ldots t_\nu\}, r, c, A, W$}
\State Fill in the values $\{t_1, \ldots, t_\nu\}$ into the table $A$, update $r$, $c$, and $W$ accordingly. 
\State For each row, fill in $A$ any uniquely determined values based on the row sum.  
\State For each column, fill in $A$ any uniquely determined values based on the column sum. 
\State If any values were filled in from the previous two steps, label them as $\{t_{\nu+1}, \ldots, t_d\}$ and update the values of $r$, $c$, $A$, and $W$ accordingly. 
\If {no new elements are filled in, i.e., $d=0$, or if $W$ is the matrix of all 1s}
	\State {\bf return} $(\{t_1, \ldots, t_\nu, \ldots, t_d\}, r, c, A, W$)
\Else
	\State {\bf return} {\tt Deterministic\_Fill}$(\{t_1,\ldots,t_d\}, r, c, A, W)$
\EndIf
\EndProcedure
\end{algorithmic}
\end{algorithm}

The rejection procedure we thus champion for approximate sampling of contingency tables is as follows. 
For each entry $(i,j)$, we calculate the set of all values which would be determined if we set that particular entry's bit equal to 0 or 1 using Procedure~{\tt Deterministic\_Fill}. 
We then sample the bit in proportion to the probability that all of the determined bits were also generated with their uniquely determined values, and apply a rejection probability similar to Equation~\eqref{fij}, except we condition on these determined values.
This is summarized in Procedure~{\tt Approximate\_Probabilistic} below. 

\ignore{
We define
\begin{align}
\nonumber & F(i,j,k, r, c,W, \{t_1,\ldots,t_\nu\}) := \\
\label{alternative:rejection} &\qquad  \Pr\left( \sum_{\ell \in I_i, \ell < j} 2\xi''_{i,\ell}(q_\ell^2, c_j, \epsilon) + 2\eta''_{i,j,i} + \sum_{\ell \in I_i, \ell > j} \xi'_{i,\ell}(q_\ell) = r_i - k \middle| \{t_1, \ldots, t_\nu\}\right) \\
\nonumber &\qquad  \ \ \ \times \Pr\left( \sum_{\ell\in J_j, \ell \leq i} 2\xi_{\ell,j}(q_j^2) + \sum_{\ell \in J_j, \ell > i} \xi_{\ell,j}(q_j) = c_j - k \middle| \{t_1 \ldots t_\nu\} \right), \qquad \qquad k \in \{0,1\}. 
 \end{align} 
Note that in general we cannot be more specific in our notation, since $t_1$ might contain a deterministic bit, or an entire entry. 
In each case it is straightforward to rewrite the explicit probability. 
}

\begin{algorithm}[H]
\label{bit_approx}
\caption*{Approximate sampling of least significant bit of entry $(i,j)$ via rejection sampling}
\begin{algorithmic}[1]
\Procedure{Approximate\_Probabilistic}{$i,j,r,c,A,W_b$}
\State $({\bf s}_0,r_0,c_0,A_0,W_0) \leftarrow$ {\tt Deterministic\_Fill}$(\{(i,j,0)\}, r, c, A, W)$
\State $({\bf s}_1, r_1,c_1,A_1,W_1) \leftarrow$ {\tt Deterministic\_Fill}$(\{(i,j,1)\}, r, c, A, W)$
\State $p_0 \leftarrow \prod_{\ell} \P(X_{i_\ell^0, j_\ell^0} = k_\ell^0) \times F(i,j,0,r_0,c_0,W_0)$
\State $p_1 \leftarrow \prod_{\ell} \P(X_{i_\ell^1, j_\ell^1} = k_\ell^1) \times F(i,j,0,r_1,c_1,W_1)$
\If{$u \leq \frac{p_0}{p_0+p_1}$}
	       \State {\bf return} $0$
	  \Else
	       \State {\bf return} $1$
	  \EndIf
\EndProcedure 
\end{algorithmic}
\end{algorithm}

Note that the $0$ in the term $F(i,j,0,r_1,c_1,W_1)$ above is \emph{not} a typo, since the row sums and column sums are assumed to be updated in the {\tt Deterministic\_Fill} algorithm. 
In addition, we did not specify whether the random variables $X_{i,j}$ refer to the entire entry or the least significant bit, since they reflect the information contained in the $t_\ell$, $\ell=1,\ldots,\nu$. 
In principle, the $X_{i,j}$ could refer to any statistic related to entry $(i,j)$, though for us we only consider the value of the entry $(i,j)$ or its least significant bit. 

\begin{lemma}\label{deterministic_fill_lemma}
The time to evaluate Procedure {\tt Deterministic\_Fill} is $O(n^2+m^2)$. 
\end{lemma}
\begin{proof}
Each time we call the function, it iterates through all elements of the table, say twice, at cost $O(m\,n)$. 
At worst, we fill in one entire row or column at each recursive call of the function. 
Since we thus invoke the recursion at most $\max(n,m)$ times, the result follows. 
\end{proof}

\ignore{
\begin{algorithm}[H]
\label{bit_approx}
\caption*{Given entries with values to fill in, return the set of determined entries.  Here $t_\ell = (i_\ell, j_\ell, k_\ell)$ is a triple consisting of a row number, column number, and bit value, respectively. }
\begin{algorithmic}[1]
\Procedure{Deterministic\_Fill}{$\{t_1, \ldots t_\nu\}, r', c', W'$}
\Statex \ \quad \textit{Initialize local variables}
\State $s \leftarrow \{t_1, \ldots, t_\nu \}$, $r \leftarrow r'$, $c \leftarrow c'$, $W \leftarrow W'$
\Statex \qquad \quad \textit{Fill in the proposed values into the table}
\For{$\ell=1,2,\ldots, \nu$}
	\State $W(i_\ell,j_\ell) \leftarrow 1$
	\State $r_{i_\ell} \leftarrow r_{i_\ell}-k_\ell$
	\State $c_{j_\ell} \leftarrow c_{j_\ell}-k_\ell$
\EndFor
\Statex \qquad \quad \textit{Check if any row entries are uniquely determined}
\For{$\ell=1,2,\ldots, m$} 
	\State $w \leftarrow \sum_{j=1}^n W(\ell,j)$
	\If{ $r_\ell = n-w$ } 
		\State $r_\ell \leftarrow 0$
		\For{$p=j+1, j+2, \ldots, n$}
			\State $W(\ell,p) \leftarrow 1$
			\State $c_p \leftarrow c_p - 1$
			\State $s \leftarrow s \cup \{ (\ell, p, 1) \}$
		\EndFor
	\ElsIf{ $r_\ell = 0$}
		\For{$p=j+1,j+2,\ldots,n$}
			\State $W(\ell,p) \leftarrow 1$
			\State $s \leftarrow s \cup \{ (\ell, p, 0) \}$
		\EndFor
	\EndIf
\EndFor
\Statex \qquad \quad \textit{Check if any column entries are uniquely determined}
\For{$p=1,2,\ldots, n$} 
	\State $w \leftarrow \sum_{i=1}^n W(i,p)$
	\If{ $c_p = m-w$ } 
		\State $c_p \leftarrow 0$
		\For{$\ell=i+1, i+2, \ldots, m$}
			\State $W(\ell,p) \leftarrow 1$
			\State $r_\ell \leftarrow r_\ell - 1$
			\State $s \leftarrow s \cup \{ (\ell, p, 1) \}$
		\EndFor
	\ElsIf{ $c_p = 0$}
		\For{$p=j+1,j+2,\ldots,n$}
			\State $W(\ell,p) \leftarrow 1$
			\State $s \leftarrow s \cup \{ (\ell, p, 0) \}$
		\EndFor
	\EndIf
\EndFor
\If{$W = W'$}
	\State {\bf return} $(s,r,c,W)$
\Else
	\State {\bf return} {\tt Deterministic\_Fill}$(s,r,c,W)$
\EndIf
\EndProcedure 
\end{algorithmic}
\end{algorithm}
}

We now state the main algorithm championed for contingency tables, which depending on several parameters determines the cost of the algorithm and the bias. 
This algorithm serves as the main inspiration for Algorithm~\ref{latin:square:algorithm}, and is interesting in its own right. 
 
\begin{algorithm}[H]
\caption{Random sampling of nonnegative integer-valued $(r,c)$-contingency table with forced zero entries in $W$.}
\label{CT:algorithm}
\begin{algorithmic}[1]
\State Fix a procedure to sample a bit, denoted by {\tt Bit\_Sampler}.
\State $A \leftarrow 0^{m\times n}$ (i.e., $A$ is initialized to be the $m\times n$ matrix of all 0s).
\State $({\bf s}, r, c, A, W) \leftarrow$ {\tt Deterministic\_Fill}($ \{\}, r, c, A, W)$.\label{preprocess}
\State $b_{\text{max}} \leftarrow \lceil \log_2(n) \rceil$.
\For{$b = 0, 1, \ldots, b_{\text{max}}$}
\State $W_b \leftarrow W.$
\State $X \leftarrow 0^{m\times n}$.
\For {$j=1,2,\ldots,n$}
  \For{$i=1,2,\ldots,m$} 
  	\If{$W_b(i,j) \neq 1$}
  	 \State $X_{i,j} \leftarrow$ {\tt Bit\_Sampler}$(i,j,r,c,W_b)$.\label{line_reject_bit}
	 \State $({\bf s}, r, c, X, W_b) \leftarrow$ {\tt Deterministic\_Fill}($ \{(i,j,X_{i,j})\}, r, c, X, W_b)$.\label{preprocess}
	 \EndIf
   \EndFor
\EndFor
\State $A \leftarrow A + 2^b X$
\EndFor
\State {\bf return} $A$.
\end{algorithmic}
\end{algorithm}

The function we choose for {\tt Bit\_Sampler} determines the bias as well as the complexity, and we have presented three explicit possibilities: {\tt Exact}, {\tt Exact\_Probabilistic}, and {\tt Approximate\_Probabilistic}. 
The rejection function $g$ in Line~\ref{reject_g} of Procedure~{\tt Exact}, given in Equation~\eqref{equation:mn}, is stated as an expression involving counting the total number of such tables; this is the direct analog of the recursive method, and typically requires a large table of values computed via a recursion. 
It is similarly difficult to compute the rejection function $f$ in Line~\ref{reject_f} of Procedure~{\tt Exact\_Probabilistic}, i.e., the probabilistic equivalent of $g$, but which motivates our approximation heuristic. 
Procedure~{\tt Approximate\_Probabilistic} offers an alternative which is both practical and an heuristically reasoned approximation for the exact quantities. 

\ignore{
The resulting algorithm is a generalization of~\cite[Algorithm~2]{DeSalvoCT}, and in fact, it is equivalent to Algorithm~\ref{CT:algorithm} above, except it is based on rejection sampling, where one instead samples from a related distribution and applies a rejection probability which rejects samples in a proportion which yields the desired distribution. 
One simply replaces lines 12--16 in Algorithm~\ref{CT:algorithm} above with the following lines:

\begin{tabular}{l}
12: $X(i,j) \leftarrow \mbox{Bern}(c_j / (m-h_j))$ \\
13: {\bf if} $U >  \frac{f(i,j,X(i,j),r,c,W_b)}{\max_{\ell \in \{0,1\}f(i,j,\ell,r,c,W_b)}}$ {\bf then} \\
14: \qquad Goto Line~12. \\
15: {\bf endif}
\end{tabular}

The rejection function $f$ is given by Equation~\eqref{fij}. 
Again, calculating functions $g$ or $f$ become the bottleneck in any computation. 

The form of rejection function $f$, however, suggests an approximation using rejection function $F$ given by Equation~\eqref{alternative:rejection} in place of $f$. 
 The first term in Equation~\eqref{alternative:rejection} is a probability over a sum of \emph{independent} random variables with explicitly known distributions, see Section~\ref{A:1}, and as such can be computed using convolutions in space/time $O(n\, M^2)$ or fast Fourier transforms in space/time $O(n\, M \log M)$. 
Further speedups are possible since only the first few bits of the function are needed on average. 

 
It is also straightforward to adapt this approach to binary tables, where we recall that such tables have values restricted to be in $\{0,1\}$.  
A similar approach was obtained specifically for the case of binary contingency tables without specified zeros in~\cite{chen} using an approach called sequential importance sampling, where at each step one keeps track of weights associated with each placement of a 0 or 1 in each entry, and at the end of the algorithm obtains an importance sampling weight associated with each sample. 
This algorithm is ideal for the first stage of the Latin square algorithm, since we are able to take advantage of the Gale--Ryser necessary and sufficient condition for binary contingency tables~\cite{gale1957theorem, Ryser} after completion of each column; subsequent steps utilize Algorithm~\ref{binary:algorithm}, and to our knowledge there is no analogous Gale--Ryser condition applicable to binary tables with an arbitrary prescribed set of zeros. 
 }

\begin{example}
As an example, an application of~\cite[Algorithm~2]{DeSalvoCT} yielded the following set of random bits in each of the five iterations of the algorithm for the random sampling of a contingency table with row sums $(40, 30, 30, 50, 100, 50)$ and column sums $(50, 50, 50, 50, 50, 50)$.\footnote{An implementation in Mathematica took about 1 second.}

\begin{align*}\small
\hskip-.5in \left(
\begin{array}{cccccc}
 12 & 12 & 3 & 0 & 9 & 4 \\
 2 & 7 & 0 & 11 & 5 & 5 \\
 2 & 9 & 10 & 0 & 4 & 5 \\
 6 & 1 & 24 & 3 & 14 & 2 \\
 18 & 2 & 12 & 32 & 16 & 20 \\
 10 & 19 & 1 & 4 & 2 & 14
\end{array}
\right) & = \ \ \ \small
\left(
\begin{array}{cccccc}
 0 & 0 & 1 & 0 & 1 & 0 \\
 0 & 1 & 0 & 1 & 1 & 1 \\
 0 & 1 & 0 & 0 & 0 & 1 \\
 0 & 1 & 0 & 1 & 0 & 0 \\
 0 & 0 & 0 & 0 & 0 & 0 \\
 0 & 1 & 1 & 0 & 0 & 0
\end{array}
\right)  
 +2^1 \left(
\begin{array}{cccccc}
 0 & 0 & 1 & 0 & 0 & 0 \\
 1 & 1 & 0 & 1 & 0 & 0 \\
 1 & 0 & 1 & 0 & 0 & 0 \\
 1 & 0 & 0 & 1 & 1 & 1 \\
 1 & 1 & 0 & 0 & 0 & 0 \\
 1 & 1 & 0 & 0 & 1 & 1
\end{array}
\right) \\ \\
& + 2^2 \left(
\begin{array}{cccccc}
 1 & 1 & 0 & 0 & 0 & 1 \\
 0 & 1 & 0 & 0 & 1 & 1 \\
 0 & 0 & 0 & 0 & 1 & 1 \\
 1 & 0 & 0 & 0 & 1 & 0 \\
 0 & 0 & 1 & 0 & 0 & 1 \\
 0 & 0 & 0 & 1 & 0 & 1
\end{array}
\right)
 + 2^3 \left(
\begin{array}{cccccc}
 1 & 1 & 0 & 0 & 1 & 0 \\
 0 & 0 & 0 & 1 & 0 & 0 \\
 0 & 1 & 1 & 0 & 0 & 0 \\
 0 & 0 & 1 & 0 & 1 & 0 \\
 0 & 0 & 1 & 0 & 0 & 0 \\
 1 & 0 & 0 & 0 & 0 & 1
\end{array}
\right) \\ \\
& + 2^4\left(
\begin{array}{cccccc}
 0 & 0 & 0 & 0 & 0 & 0 \\
 0 & 0 & 0 & 0 & 0 & 0 \\
 0 & 0 & 0 & 0 & 0 & 0 \\
 0 & 0 & 1 & 0 & 0 & 0 \\
 1 & 0 & 0 & 0 & 1 & 1 \\
 0 & 1 & 0 & 0 & 0 & 0
\end{array}
\right) 
 + 2^5 \left(
\begin{array}{cccccc}
 0 & 0 & 0 & 0 & 0 & 0 \\
 0 & 0 & 0 & 0 & 0 & 0 \\
 0 & 0 & 0 & 0 & 0 & 0 \\
 0 & 0 & 0 & 0 & 0 & 0 \\
 0 & 0 & 0 & 1 & 0 & 0 \\
 0 & 0 & 0 & 0 & 0 & 0
\end{array}
\right).
\end{align*}

\end{example}

\begin{example}
Consider the following example, a $3 \times 4$ nonnegative integer-valued $(r,c)$-contingency table with row sums $r = (20, 14, 18, 27)$ and column sums $c = (13, 56, 10)$.  
Furthermore, we force entries $(1,3)$, $(2,1)$, $(3,2)$, $(3,3)$, $(3,4)$ to be 0.  The first step is Line~\ref{preprocess}, which transforms the initial sampling problem into the following:
\[\hskip-.5in
\begin{array}{ccccccc}
\begin{array}{cc}
\begin{array}{|c|c|c|c|} \hline
 & & \blacksquare & \\ \hline
 \blacksquare & & & \\ \hline
 & \blacksquare & \blacksquare & \blacksquare \\ \hline
\end{array} & 
\begin{array}{c} 10 \\ 56 \\ 13 \end{array}  \\
\begin{array}{cccc} 20 & 14 & 18 & 27 \end{array}
\end{array} & 
\longrightarrow &
\begin{array}{cc}
\begin{array}{|c|c|c|c|} \hline
7 & & \blacksquare & \\ \hline
 \blacksquare & & 18 & \\ \hline
13 & \blacksquare & \blacksquare & \blacksquare \\ \hline
\end{array} & 
\begin{array}{c} 10 \\ 56 \\ 13 \end{array}  \\
\begin{array}{cccc} 20 & 14 & 18 & 27 \end{array}
\end{array} &
\longrightarrow & \\ \\
\begin{array}{cc}
\begin{array}{|c|c|c|c|} \hline
\blacksquare & & \blacksquare & \\ \hline
 \blacksquare & & \blacksquare & \\ \hline
\blacksquare & \blacksquare & \blacksquare & \blacksquare \\ \hline
\end{array} & 
\begin{array}{c} 3 \\ 38 \\ 0 \end{array}  \\
\begin{array}{cccc} 0 & 14 & 0 & 27 \end{array}
\end{array} &
\longrightarrow &
\begin{array}{cc}
\begin{array}{|c|c|} \hline
\ \  & \ \  \\ \hline
\ \  & \ \  \\ \hline
\end{array} & 
\begin{array}{c} 3 \\ 38 \end{array}  \\
\begin{array}{cccc} 14 & 27  \end{array}
\end{array} 
\end{array}
\]
Now since this reduces to a standard $2\times 2$ contingency table, we can apply a standard algorithm or continue with Algorithm~\ref{CT:algorithm}. 
We suggest in particular the \emph{exact} sampling algorithm for $2 \times n$ tables given in~\cite[Algorithm~3]{DeSalvoCT}, which for this particular example would sample a uniform number in the set $\{0, 1, 2, 3\}$ and then uniquely fill in the remaining entries. 

\end{example}

 \subsection{Binary contingency tables}
\label{A:binary}
One can formulate an exact sampling rejection function for binary contingency tables as in the previous section, although the approximation algorithm is again the one we champion due to practical computing reasons. 
It is straightforward but notationally messy to adapt and write out completely the analogous exact sampling procedures in this setting, and so we focus solely on our suggested approximate sampling approach. 

\ignore{
Let $\BC(r,c,W)$ denote the number of binary $(r,c)$-contingency tables with entry $(i,j)$ forced to be 0 if the $(i,j)$th entry of $W$ is 1. 
Let $\BC(r, c, W; t_1, \ldots, t_r)$, where each $t_\ell$, $\ell = 1, \ldots, r$ denotes an entry in the table, say $t_\ell = (i_\ell, j_\ell)$, denote the number of binary $(r,c)$-contingency tables as before with the additional assumption that the entries $t_1, \ldots t_r$ are also forced to be 0. 

Let $h_j := \sum_{i=1}^m W_{i,j}$ denote the number of non-zero elements in column $j$, $j=1,2,\ldots,n$.
Let $q_j := \frac{c_j}{m-h_j}$, and 
let $y_j := q_j^{-1}$. 

For each $1 \leq i \leq m$ and $ 1 \leq j \leq n$, let $t_{1}, \ldots, t_{r}$ denote the set of entries which are determined by the line sum conditions after sampling of entry $(i,j)$, and let $b_{1}(k), \ldots, b_{r}(k)$ denote the deterministic value of the entry in order to satisfy the row and column constraints given that entry $(i,j)$ has value $k$, where $k \in \{0,1\}$. 
Then define
\begin{align*}
A & \equiv \BC\left(\begin{array}{l} (\ldots,r_i-k,\ldots, r_{i_1} - b_{1}(k),\ldots, r_{i_r} - b_{r}(k)), \\ (\ldots,c_j-k, \ldots, c_{j_1} - b_{1}(k),\ldots, c_{j_r} - b_{r}(k)), \\ W,\ \mathcal{O}_{i,j};\  t_1, \ldots, t_r \end{array} \right)\,\cdot \, \prod_{\ell=1}^r y_{j_\ell}^{b_{\ell}(k)},  \\
B & \equiv \BC\left(\begin{array}{l} (\ldots,r_i,\ldots, r_{i_1} - b_{1}(0),\ldots, r_{i_r} - b_{r}(0)), \\ (\ldots,c_j, \ldots, c_{j_1} - b_{1}(0),\ldots, c_{j_r} - b_{r}(0)), \\ W,\ \mathcal{O}_{i,j}; \ t_1, \ldots, t_r \end{array} \right)\,\cdot \, \prod_{\ell=1}^r y_{j_\ell}^{b_{\ell}(0)}, \\
C & \equiv \BC\left(\begin{array}{l} (\ldots,r_i-1,\ldots, r_{i_1} - b_{1}(1),\ldots, r_{i_r} - b_{r}(1)), \\ (\ldots,c_j-1, \ldots, c_{j_1} - b_{1}(1),\ldots, c_{j_r} - b_{r}(1)), \\ W,\ \mathcal{O}_{i,j};\  t_1, \ldots, t_r \end{array} \right)\,\cdot \, \prod_{\ell=1}^r y_{j_\ell}^{b_{\ell}(1)}, 
\end{align*}
\begin{equation}\label{binary:equation:mn} h(i,j,k,r,c,W) :=  \frac{A}{B+C}. \end{equation}
}
\ignore{
\begin{algorithm}[H]
\caption{Generation of an exactly uniform binary $(r,c)$-contingency table with forced zero entries in W.}
\label{binary:algorithm}
\begin{algorithmic}[1]
\State {\bf Input: } $r = (r_1, \ldots, r_m)$, $c = (c_1, \ldots, c_n)$, and $W \in \{0,1\}^{m\times n}$.
\State {\bf Output: } $(r,c)$-binary contingency table with entires in $W$ forced to be 0.

\State Let $J_j(W) \leftarrow \{1,\ldots, n\} \setminus \{ i : W(i,j) = 1\}$ for each $j = 1,2,\ldots,n$.
\State Let $k_j$ denote the last entry in $J_j(W)$, $j=1,2,\ldots,n$. 
\For {$j=1,2,\ldots,n$}
  \For{each $i \in J_j\setminus k_j$}
     \State $q_j \leftarrow \frac{c_j}{(m-h_j)+c_j}$
  	  \If{ $U \leq  h(i,j,0,r,c,W) $}
	       \State $X(i,j) \leftarrow 0$.
	  \Else
	       \State $X(i,j) \leftarrow 1$.
	  \EndIf
	  	 \State $W(i,j) \leftarrow 1.$
	 \State Fill in all deterministic bits based on parity of constraints. 
	 \State $W(r,s) \leftarrow 1$ where $(r,s)$ runs through all entries which were filled in deterministically.
	 \State Update $r, c, J_j, k_j.$ 

    \EndFor
\EndFor
\State {\bf return} $X$.
\end{algorithmic}
\end{algorithm}

Next we present a more probabilistic approach, which yields a natural approximate sampling algorithm. 
}

Let $B_j(q)$, $j=1,\ldots,n$, denote independent Bernoulli random variables with success probability $1-q$.  
We have 
\begin{equation}\label{g:binary} 
b(n,{\bf q},k)  := \P\left( \sum_{j=1}^n B_j(q_j) = k\right) \nonumber = \frac{1}{n+1}\sum_{\ell = 0}^n C^{-\ell k}\prod_{j=1}^n \left(1+(C^\ell-1)(1-q_j)\right),  
\end{equation}
where $C = \exp\left(\frac{2\pi i}{n+1}\right)$ is a $(n+1)$th root of unity.  
\ignore{Define 
\[ d_{n,k} := \frac{\max\left(b(n,{\bf q}, {\bf c}, k), b(n,{\bf q}, {\bf c}, k-1)\right)}{b(n-1,{\bf q}, {\bf c}, k))}, \]
and $d := \max_{n,k} d_{n,k}$. }
The expression in~\eqref{g:binary} is a numerically stable way to evaluate the convolution of a collection of independent Bernoulli random variables using a fast Fourier transform, see for example~\cite{PoissonBinomial}. 
Next, we consider an $m\times n$ matrix $W$ which contains values in $\{0,1\}$.  A value 1 in entry $(i,j)$ of $W$ implies that this entry is forced to be zero. 
Given an index set $J \subset \{1,2,\ldots, n\}$, we define similarly
\[ b(n,{\bf q},k,J) := \P\left( \sum_{j\in J} B_j(q_j) = k\right). \]


For each $i=1,2,\ldots,m$, let $I_i \leftarrow \{1,\ldots, n\} \setminus \{j : W(i,j) = 1\}$ denote the set of entries in row $i$ which can be non-zero; and for each $j=1,2,\ldots,n$, let $J_j \leftarrow \{1,\ldots, n\} \setminus \{ i : W(i,j) = 1\}$ denote the set of entries in column $j$ which can be non-zero. 
Also let ${\bf q} = (c_1/(m-h_1), c_2/(m-h_2),\ldots, c_n/(m-h_n))$ and let ${\bf \nu_j} = (c_j/(m-h_j), c_j/(m-h_j), \ldots, c_j/(m-h_j))$. 
Then we define for $i \in I_i$ and $j \in J_j$
\begin{equation} 
\label{binary:rejection}
H(i,j,k,r,c,W) := b(m, {\bf q}, r_i - k, I_i) \times b(n,{\bf \nu_j}, c_j-k, J_j), 
\end{equation}
and $0$ otherwise. 
We therefore define the following procedure to generates bits, noting that the wording of the {\tt Deterministic\_Fill} function defined in the previous section is equally applicable in this setting.

\begin{algorithm}[H]
\label{bit_exact}
\caption*{Approximate sampling of entry $(i,j)$ of a binary contingency table}
\begin{algorithmic}[1]
\Procedure{Approximate\_Probabilistic\_Binary}{$i,j,r,c,A,W$}
\State $({\bf s}_0,r_0,c_0,A_0,W_0) \leftarrow$ {\tt Deterministic\_Fill}$(\{(i,j,0)\}, r, c, A, W)$
\State $({\bf s}_1, r_1,c_1,A_1,W_1) \leftarrow$ {\tt Deterministic\_Fill}$(\{(i,j,1)\}, r, c, A, W)$
\State $p_0 \leftarrow \prod_{\ell} \P(X_{i_\ell^0, j_\ell^0} = k_\ell^0) \times H(i,j,0,r_0,c_0,W_0)$
\State $p_1 \leftarrow \prod_{\ell} \P(X_{i_\ell^1, j_\ell^1} = k_\ell^1) \times H(i,j,0,r_1,c_1,W_1)$
\If{$u \leq \frac{p_0}{p_0+p_1}$}
	       \State {\bf return} $0$
	  \Else
	       \State {\bf return} $1$
	  \EndIf
\EndProcedure 
\end{algorithmic}
\end{algorithm}

\ignore{
For a given set of row sums $r = (r_1, \ldots, r_m)$ and column sums $c = (c_1, \ldots, c_n)$, let $\B\C(r,c)$ denote the number of $(r,c)$-binary contingency tables. 
Let $\mathcal{O}$ denote an $m\times n$ matrix with entries in $\{0,1\}$.  
For any set of row sums and column sums $(r,c)$, let $\B\C(r,c,\mathcal{O})$ denote the number of $(r,c)$-contingency tables with entry $(i,j)$ forced to be 0 if the $(i,j)$th entry of $\mathcal{O}$ is 1, and no restriction otherwise. 
Let $\mathcal{O}_{i,j}$ denote the matrix which has entries with value $1$ in the first $j-1$ columns, and entries with value $1$ in the first $i$ rows of column $j$, and entries with value 0 otherwise. 

Define for $1 \leq i \leq m-2, 1 \leq j \leq n-2$, 
\begin{equation}\label{rejection:ij}\small
 f(i,j,k,r,c) :=  \frac{\B\C\left( \begin{array}{l}(\ldots,r_{i}-k,\ldots), \\ (\ldots,c_{j}-k,\ldots), \\ \mathcal{O}_{i,j}\end{array} \right)}{\B\C\left( \begin{array}{l}(\ldots,r_{i}-1,\ldots), \\ (\ldots,c_{j}-1,\ldots), \\ \mathcal{O}_{i,j}\end{array} \right) + \B\C\left( \begin{array}{l}(\ldots,r_{i},\ldots), \\ (\ldots,c_{j},\ldots), \\ \mathcal{O}_{i,j}\end{array} \right)}.
\end{equation}
}

\ignore{
\begin{algorithm}[H]
\caption{Generation of an approximately uniform binary $(r,c)$-contingency table with forced zero entries in W.}
\label{binary:algorithm:approx}
\begin{algorithmic}[1]
\State {\bf Input: } $r = (r_1, \ldots, r_m)$, $c = (c_1, \ldots, c_n)$, and $W \in \{0,1\}^{m\times n}$.
\State {\bf Output: } $(r,c)$-binary contingency table with entires in $W$ forced to be 0.

\State Let $J_j(W) \leftarrow \{1,\ldots, n\} \setminus \{ i : W(i,j) = 1\}$ for each $j = 1,2,\ldots,n$.
\State Let $k_j$ denote the last entry in $J_j(W)$, $j=1,2,\ldots,n$. 
\For {$j=1,2,\ldots,n$}
  \For{each $i \in J_j\setminus k_j$}
     \State $q_j \leftarrow \frac{c_j}{(m-h_j)+c_j}$
     \State $X(i,j) \leftarrow \mbox{Bern}(q_j / (1+q_j))$. \label{Bern}
  	  \If{ $U > \frac{H(i,j,X(i,j),r,c,W)}{\max_{\ell \in \{0,1\}H(i,j,\ell,r,c,W)}} $}
	       \State Goto Line~\ref{Bern}.
	  \EndIf
	  	 \State $W(i,j) \leftarrow 1.$
	 \State Fill in all deterministic bits based on parity of constraints. 
	 \State $W(r,s) \leftarrow 1$ where $(r,s)$ runs through all entries which were filled in deterministically.
	 \State Update $r, c, J_j, k_j.$ 

    \EndFor
\EndFor
\State {\bf return} $X$.
\end{algorithmic}
\end{algorithm}
}

\begin{lemma}\label{binary_lemma}
Procedure {\tt Approximate\_Probabilistic\_Binary} has an arithmetic cost of $O(n^2+m^2)$
\end{lemma}
\begin{proof}
By Lemma~\ref{deterministic_fill_lemma}, each of the two calls to Procedure~{\tt Deterministic\_Fill} take at most $O(n^2+m^2)$. 
In addition, the cost to evaluate $H$ via Equation~\eqref{binary:rejection} and Equation~\eqref{g:binary} is also at most $O(n^2+m^2)$. 
\end{proof}

We thus define the algorithm for approximate sampling of binary $(r,c)$-contingency tables as follows. 

\begin{algorithm}[H]
\caption{Generation of an approximately uniform binary $(r,c)$-contingency table with forced zero entries in W.}
\label{binary:algorithm:approx}
\begin{algorithmic}[1]
\State {\bf Input: } $r = (r_1, \ldots, r_m)$, $c = (c_1, \ldots, c_n)$, and $W \in \{0,1\}^{m\times n}$.
\State {\bf Output: } binary $(r,c)$-contingency table with entries in $W$ forced to be 0.
\State $X \leftarrow 0^{m\times n}$ (i.e., $X$ is initialized to be the $m\times n$ matrix of all 0s).
\State $({\bf s}, r, c, X, W) \leftarrow$ {\tt Deterministic\_Fill}($ \{\}, r, c, X, W)$.\label{preprocess}
\For {$j=1,2,\ldots,n$}
  \For{$i=1,2,\ldots,m$} 
  	\If{$W(i,j) \neq 1$}
  	 \State $X_{i,j} \leftarrow$ {\tt Approximate\_Probabilistic\_Binary}$(i,j,r,c,W)$.\label{line_reject_bit_binary}
	 \State $({\bf s}, r, c, X, W) \leftarrow$ {\tt Deterministic\_Fill}($ \{(i,j,X_{i,j})\}, r, c, X, W)$.\label{preprocess}
	 \EndIf
   \EndFor
\EndFor
\State {\bf return} $X$.\end{algorithmic}
\end{algorithm}

\subsection{Latin squares of order $n$}
\ignore{
Suppose $b = (b_1, \ldots, b_k)\in \{0,1\}^k$ is some combination of $k$ 0s and 1s, for any $k \geq 1$. 
Given also positive integer $n$, let $n_b$ denote a sequence of ceiling and floor functions, each followed by a division by 2, applied to $\lceil n/2 \rceil$.  For example, 
\[ n_{101} = \lceil \lfloor \lceil \lceil n/2 \rceil /2 \rceil / 2\rfloor / 2 \rceil. \]

We now define elements $t_b$ inductively.  
Start with $t^{(0)}$, an $m \times n$ matrix of $0$s. 
Let $t_0$ denote an $m \times n$ matrix, conditional on $t^{(0)}$, which is 1 wherever $t$ is 0, and 0 otherwise.
Similarly, let $t_1$ denote an $m \times n$ matrix, conditional on $t$, which is 1 whenever $t$ is 1, and 0 otherwise. 
Let $t^{(1)}$ denote an $m\times n$ matrix with each entry consisting of tuple $(g_0, g_1)$, where $g_0$ is the corresponding element in $t_0$, and $g_1$ is the corresponding element in $t_1$. 
Now define $t^{(b)}$, $b$ of size $k$, conditional on $t^{(b)}$, $b$ of size $k-1$, to be an $m\times n$ matrix which is 1 whenever 

Similarly, let $W_b$ denote a matrix in $\{0,1\}^{m\times n}$ which is conditional on 
}

We now state a complete version of our main algorithm for random sampling of Latin squares. 

\begin{algorithm}[H]
\caption{Generation of an approximately uniformly random Latin square of order $n$.}
\label{latin:square:algorithm}
\begin{algorithmic}[1]
\State $r \leftarrow (\lfloor n/2 \rfloor, \ldots, \lfloor n/2 \rfloor)$.
\State $c \leftarrow (\lfloor n/2 \rfloor, \ldots, \lfloor n/2 \rfloor)$. 
\State Let $W$ denote the $n\times n$ matrix with all 0 entries.
\State $t \leftarrow$ the output from Algorithm~\ref{binary:algorithm:approx} with input $(r,c,W)$.
\For {$i=1,2,\ldots, \lceil \log_2(n)\rceil$}
        \State $r \leftarrow \lfloor r/2\rfloor$.
        \State $c \leftarrow \lfloor c/2 \rfloor$. 
    \For{$ b = 0, 1, \ldots, 2^{i}-1$} 
        \State Let $W_b(i,j)$ be 0 if $t(i,j) = b$, and 1 otherwise, for all $1 \leq i \leq m, 1 \leq j \leq n$.
        \State $t_b \leftarrow$ the output from Algorithm~\ref{binary:algorithm:approx} with input $(r, c, W_b)$.\label{rejection_goes_here}
        \State $t \leftarrow t + 2^i t_b$
    \EndFor
\EndFor
\State {\bf return} $t+1$
\end{algorithmic}
\end{algorithm}

\ignore{
\begin{remark}
It bears repeating that while a priori this algorithm appears to be a sum over an exponentially growing number of tables, and indeed for exposition purposes we have chosen this form in order to clearly convey the procedure, each of the $2^k$ calls to Algorithm~\ref{binary:algorithm} inside the innermost for loop can be performed with one loop over each entry in the table, being careful to apply the correct rejection function for each entry $(i,j)$ based on previously observed bits in entry $(i,j)$. 
This makes the required time $O(n^2)$ times the cost to evaluate and pass the rejection function threshold in Algorithm~\ref{binary:algorithm}. 
\end{remark}}

\begin{remark}
It is easy to adapt this algorithm for the random sampling of $n^2 \times n^2$ Sudoku matrices, with an extra rejection due to the local block constraints; see for example~\cite{DeSalvoNewton} for the precise definition.  
One could start with Algorithm~\ref{latin:square:algorithm} and perform a hard rejection at each iteration until the binary contingency table generated also satisfies the block constraints, or one could adapt the rejection probability to more effectively target the block constraints. 
\end{remark}

 
\subsection{Algorithmic costs and a word of caution}

We are unable at present to formulate precise runtime estimates, only to say we have implemented the approximation algorithms in practice and they exhibit roughly the runtime estimates we show below, with one major caveat. 
Part of the difficulty in precisely costing the algorithms is that, for the exact sampling algorithms, they rely on computing numerical quantities for which no known polynomial time algorithm is known.  
For the approximation algorithms, we have not yet undertaken a full analysis of the bias, and in fact, we caution that the algorithm can potentially reach a state for which the algorithm cannot terminate due to unsatisfiable constraints. 
In those cases it is up to the practitioner whether the algorithm ought to halt altogether and restart, or whether some partially completed table should be salvaged; again, since we have not performed a detailed analysis of the bias in these algorithms, we have not explored this matter further, except to note that the algorithm is surprisingly effective when it avoids such unsalvageable states. 

In each of the estimates below, we assume $m$ denotes the number of rows and $n$ denotes the number of columns in a contingency table, and that $W$ is some $m \times n$ matrix with a $1$ in the $(i,j)$-th entry if the $(i,j)$-th entry of the contingency table is forced to be 0. 
We also define 
\[ R := \max(r_1, \ldots, r_m), \qquad C := \max(c_1, \ldots, c_n), \qquad M := \max(R,C). \]

For Algorithm~\ref{CT:algorithm}, with Procedure~{\tt Approximate\_Probabilistic} used for the {\tt Bit\_Sampler} routine, it is easy to see that barring the encounter of a state which cannot be completed, the expected number of random bits required is $O(n\,m\,\log(M))$, one for each visit to each entry per bit-level. 
For the arithmetic cost, each call to Procedure~{\tt Approximate\_Probabilistic} involves $O(R \log R)$ arithmetic operations, since the evaluation of the function $F$ can be performed in a numerically stable manner using fast Fourier transforms. 
In addition, each call to Procedure~{\tt Deterministic\_Fill} involves $O(m^2+n^2)$ arithmetic operations by Lemma~\ref{deterministic_fill_lemma}. 
Thus, this corresponds to an informal arithmetic cost which is $O((m^2+n^2) m\, n\, R\, \log(R) \log(M))$ in the best case scenario. 

For Algorithm~\ref{binary:algorithm:approx}, it is easy to see that the expected number of random bits is $O(n\, m)$, with $O((n^2+m^2)m\,n)$ arithmetic operations, since we iterate through each of the $n\, m$ entries at most once, each time calling Procedure~{\tt Approximate\_Probabilistic\_Binary}. 

Finally, for Algorithm~\ref{latin:square:algorithm}, the expected number of random bits required is $O(n^2 \log(n))$ with $O(n^4 \log n)$ arithmetic operations.

\ignore{
\begin{theorem}\label{ct_theorem_tables}
Consider Algorithm~\ref{CT:algorithm} with Procedure~{\tt Approximate\_Probabilistic} used for the {\tt Bit\_Sampler} routine. 
The expected number of random bits required for the algorithm to run to completion is $O( n\, m\, \log(M))$, with $O((m^2+n^2) m\, n\, R\, \log(R) \log(M))$ arithmetic operations.
\end{theorem}
\begin{proof}
Each call to Procedure~{\tt Approximate\_Probabilistic} involves $O(R \log R)$ arithmetic operations, since the evaluation of the function $F$ can be performed in a numerically stable manner using fast Fourier transforms. 
In addition, each call to Procedure~{\tt Deterministic\_Fill} involves $O(m^2+n^2)$ arithmetic operations by Lemma~\ref{deterministic_fill_lemma}. 
Finally, since we call these procedures on each of the $m\, n$ entries, and iterate at most $\log(M)$ times, the result follows. 
\end{proof}

\begin{theorem}\label{binary_ct_theorem_tables}
For Algorithm~\ref{binary:algorithm:approx}, the expected number of random bits required is $O(n\, m)$, with $O((n^2+m^2)m\,n)$ arithmetic operations.
\end{theorem}
\begin{proof}
Similarly as in the Proof of Theorem~\ref{ct_theorem_tables}, we iterate through each of the $n\, m$ entries at most once, each time calling Procedure~{\tt Approximate\_Probabilistic\_Binary}. 
By Lemma~\ref{binary_lemma}, this has cost $O(n^2+m^2)$, hence the result follows. 
\end{proof}

\begin{theorem}
For Algorithm~\ref{latin:square:algorithm}, the expected number of random bits required is $O(n^2 \log(n))$ with $O(n^4 \log n)$ arithmetic operations. 
\end{theorem}
\begin{proof}
We iterate through each of the $n^2$ entries at most $\log(n)$ times, due to the disjoint nature of the divide-and-conquer strategy on the entries.  
At each entry we potentially call Procedure~{\tt Deterministic\_Fill}, which is $O(n^2)$ by Lemma~\ref{deterministic_fill_lemma}, and hence the result follows. 
\end{proof}
}

\ignore{
\subsection{Latin squares of order $n$}
\ignore{
Suppose $b = (b_1, \ldots, b_k)\in \{0,1\}^k$ is some combination of $k$ 0s and 1s, for any $k \geq 1$. 
Given also positive integer $n$, let $n_b$ denote a sequence of ceiling and floor functions, each followed by a division by 2, applied to $\lceil n/2 \rceil$.  For example, 
\[ n_{101} = \lceil \lfloor \lceil \lceil n/2 \rceil /2 \rceil / 2\rfloor / 2 \rceil. \]

We now define elements $t_b$ inductively.  
Start with $t^{(0)}$, an $m \times n$ matrix of $0$s. 
Let $t_0$ denote an $m \times n$ matrix, conditional on $t^{(0)}$, which is 1 wherever $t$ is 0, and 0 otherwise.
Similarly, let $t_1$ denote an $m \times n$ matrix, conditional on $t$, which is 1 whenever $t$ is 1, and 0 otherwise. 
Let $t^{(1)}$ denote an $m\times n$ matrix with each entry consisting of tuple $(g_0, g_1)$, where $g_0$ is the corresponding element in $t_0$, and $g_1$ is the corresponding element in $t_1$. 
Now define $t^{(b)}$, $b$ of size $k$, conditional on $t^{(b)}$, $b$ of size $k-1$, to be an $m\times n$ matrix which is 1 whenever 

Similarly, let $W_b$ denote a matrix in $\{0,1\}^{m\times n}$ which is conditional on 
}

We now state our main algorithm for random sampling of Latin squares, presented to be as simple and concise to follow as possible. 
As stated, it is a priori the sum of an exponential number of tables; a fortiori, the non-zero entries of these tables form a partition of all $n^2$ entries, and so one would implement the algorithm below to switch between the various rejection functions and thus avoid an exponential number of calls to Algorithm~\ref{binary:algorithm}. 

\begin{algorithm}[H]
\caption{Generation of an approximately uniformly random Latin square of order $n$.}
\label{latin:square:algorithm}
\begin{algorithmic}[1]
\State $r \leftarrow (\lfloor n/2 \rfloor, \ldots, \lfloor n/2 \rfloor)$.
\State $c \leftarrow (\lfloor n/2 \rfloor, \ldots, \lfloor n/2 \rfloor)$. 
\State Let $W$ denote the $n\times n$ matrix with all 0 entries.
\State $t \leftarrow$ the output from Algorithm~\ref{binary:algorithm} with input $(r,c,W)$.
\For {$i=1,2,\ldots, \lceil \log_2(n)\rceil$}
        \State $r \leftarrow \lfloor r/2\rfloor$.
        \State $c \leftarrow \lfloor c/2 \rfloor$. 
    \For{$ b = 0, 1, \ldots, 2^{i}-1$} 
        \State Let $W_b(i,j)$ be 0 if $t(i,j) = b$, and 1 otherwise, for all $1 \leq i \leq m, 1 \leq j \leq n$.
        \State $t_b \leftarrow$ the output from Algorithm~\ref{binary:algorithm} with input $(r, c, W_b)$.
        \State $t \leftarrow t + 2^i t_b$
    \EndFor
\EndFor
\State {\bf return} t+1
\end{algorithmic}
\end{algorithm}

\begin{remark}
It bears repeating that while a priori this algorithm appears to be a sum over an exponentially growing number of tables, and indeed for exposition purposes we have chosen this form in order to clearly convey the procedure, each of the $2^k$ calls to Algorithm~\ref{binary:algorithm} inside the innermost for loop can be performed with one loop over each entry in the table, being careful to apply the correct rejection function for each entry $(i,j)$ based on previously observed bits in entry $(i,j)$. 
This makes the required time $O(n^2)$ times the cost to evaluate and pass the rejection function threshold in Algorithm~\ref{binary:algorithm}. 
\end{remark}

\begin{remark}
To make this an exact sampling algorithm one would need to apply an additional rejection after each call to Algorithm~\ref{binary:algorithm}. 
We have suppressed this rejection, as it requires enumerative estimates which are not easily accessible. 
\end{remark}
}

\ignore{
\section{Proofs}

\subsection{Proof of Theorem~\ref{exact:theorem:tables}}

We also have an analogous lemma for the quantitative bounds, cf.~\cite[Lemma~2.8]{DeSalvoCT}. 
\begin{lemma}\label{bound}
Any algorithm which uses $\L\left(\Bern\left(\frac{q_j}{1+q_j}\right)\right)$, where $q_j = \frac{c_j}{m-h_j + c_j}$, as the surrogate distribution for $\L(\epsilon_{i,j}|E)$ in rejection sampling, where $\epsilon_{i,j}$ denote the least significant bit of $X_{i,j}$, assuming each outcome in $\{0,1\}$ has a positive probability of occurring, accepts a bit with an incorrect proportion bounded by at most $m+2$, where $m$ is the number of rows. 
\end{lemma}
\begin{proof}
Let $B$ be distributed as $\L\left(\Bern\left(\frac{q_j}{1+q_j}\right)\right)$.  We have
\[ \frac{q_j}{1+q_j} = \frac{c_j}{m-h_j+2c_j}, \]
and so 
\[
\begin{array}{rcl}
 \frac{1}{m-h_j+2} & < \Pr(B = 1) < & \frac{1}{2}  \\
 \frac{1}{2}  & < \Pr(B=0) < & \frac{m-h_j+1}{m-h_j+2}.
\end{array}
\]
We are able to scale up by the larger point probability, i.e., the event $\{B=0\}$, which means we accept a generated bit of 1 with probability 1, for which we would have to wait at most an expected $m-h_j +2 \leq m+2$ iterations.  Thus, at worst we accept a bit $m+2$ times more likely in one state than the other. 
\end{proof}

All that remains is the costing estimates. 
For Algorithm~\ref{CT:algorithm}, there is no rejection, and so the cost in terms of the expected number of random bits required is simply the cost of visiting all $s$ entries at most $\log M$ times to assign a bit.

For Algorithm~\ref{CT:algorithm:approx}, the rejection adds at most $O(m)$ iterations for each entry, which implies $O(s\, m\, \log M)$ expected random bits required. 
To calculate the rejection probabilities is a convolution of at most $n$ independent random variables.  Each convolution costs a priori $O(M^2)$ directly, or $O(M\, \log M)$ using an FFT.  Repeating this $n$ times gives $O(n\, M\, \log M)$ for each entry, and so the total arithmetic cost is $O(s\, n\, M\, \log M)$. 

\subsection{Proof of Theorem~\ref{binary:theorem:tables}}

We start with a similar treatment as in the previous section, except instead of taking $X_{i,j}$ to be geometrically distributed, with parameter $p_{i,j} = 1-\alpha_i \beta_j$, we instead take $X_{i,j}$ to be Bernoulli with parameter $\frac{q_{i,j}}{1+q_{i,j}}$, where $q_{i,j} = 1-p_{i,j}$. 
In fact, this is precisely the distribution of the least significant bit utilized in Algorithm~\ref{CT:algorithm:approx}. 

\begin{lemma}
\label{lemma:uniform}
Let $\X_W=(X_{ij})_{1 \leq i \leq m, 1 \leq j \leq n}$ denote a collection of independent Bernoulli random variables with parameters $p_{ij}$, such that $W(i,j) = 1$ implies $p_{i,j} = 0$, and otherwise $p_{ij}$ has the form $p_{ij} = \frac{\alpha_i \beta_j}{1+\alpha_i \beta_j}$. Then $\X_W$ is uniform restricted to binary $(r,c)$-contingency tables with zeros in entries indicated by $W$.
\end{lemma}

The proof is straightforward, as is the analogous result for Lemma~\ref{bound} and Lemma~\ref{expectation}. 

The proof of uniformity of Algorithm~\ref{binary:algorithm} follows again almost verbatim from~\cite[Section~4]{DeSalvoCT}, aside from straightforward modifications of boundary conditions, summing over index sets $j \in I_i$ rather than $j=1,2,\ldots,n$, and replacing the role of $\Sigma(r,c,\mathcal{O})$ with $\BC(r,c,W)$.

The costing estimates are also straightforward.  
For Algorithm~\ref{binary:algorithm}, there is no rejection, and so the cost in terms of the expected number of random bits required is simply the cost of visiting all $s$ entries once.

For Algorithm~\ref{binary:algorithm:approx}, the rejection adds at most $O(m)$ iterations for each entry, which implies $O(s\, m)$ expected random bits required. 
To calculate the rejection probabilities is again a convolution of at most $n$ independent random variables.  
Each convolution costs a priori $O(M^2)$ directly, or $O(M\, \log M)$ using an FFT.  
Repeating this $n$ times gives $O(n\, M\, \log M)$ for each entry, and so the total arithmetic cost is $O(s\, n\, M\, \log M)$. 
Since this is a binary contingency table, we must have $M \leq \max(m,n)$, although we prefer to keep this as a separate variable to highlight the case of sparse tables.

\subsection{Proof of Theorem~\ref{latin:square:theorem}}

The fact that a Latin square can be decomposed into its corresponding bits is straightforward, as is the connection with collections of binary $(r,c)$-contingency tables with restrictions. 

It is not a priori obvious whether the bits at different levels are independent, and whether certain configurations of binary tables at a given level can potentially be completed by a larger number of Latin squares than other configurations. 
For $n=6$, we are able to provide a negative answer, although it would be interesting to explore how non-uniform the distribution truly is. 
According to the OEIS sequence~A058527~\cite{OEIS}, the number of $6 \times 6$ binary contingency tables with row sums and column sums equal to 3 is 297200, which does not divide the number of Latin squares of order $6$, which is 812851200, see for example~\cite{McKayWanless}.  Thus, some of the  configurations of binary tables yield a different number of completable Latin squares. 
If one could bound the range of possible completions given such a binary contingency table, then one could obtain quantitative bounds on the number of Latin squares. 



Next we consider the arithmetic cost.
The algorithm requires visiting each entry at most $\lceil \log_2(n)) \rceil$ times.  
Let us consider the first time entry $(1,1)$ is visited.  
The rejection probability given in Equation~\eqref{binary:rejection} is the product of two convolutions of at most $n$ elements each, for a total of $2\times O\left(n\, \frac{n}{2}\, \log \frac{n}{2}\right)$. 
The second time entry $(1,1)$ is visited, the cost is $O\left(\frac{n}{2} \frac{n}{4} \log \frac{n}{4}\right)$. 
Summing over at most $\lceil \log_2(n)\rceil$ visits, the total cost associated with entry $(1,1)$ is $O(n^2)$. 

Consider the entry $(i,j)$.  The convolutions in Equation~\eqref{binary:rejection} have an initial cost of 
\[O\left( (n-i) \frac{n-i}{2} \log \frac{n-i}{2} + (n-j) \frac{n-j}{2}\log \frac{n-j}{2}\right) = O\left(n\, \frac{n}{2} \log \frac{n}{2}\right). \]

Thus, summing over all entries $O(\log n)$ times we have at most $O(n^4 \log(n))$ arithmetic operations total. 

The expected number of random bits comes from the rejection probabilities at each step. 
Consider the first step, which is the random sampling of binary $(\lfloor n/2\rfloor ,\lfloor n/2\rfloor )$-contingency tables. 
By the previous section, this costs $O(n^3)$, and we repeat this at most $O(\log n)$ times. 
}

\ignore{
\section{Other extensions}
\label{continuous}
One could also consider, e.g., tables with continuous--valued entries. 
In this case, the conditioning event is more delicate, as we can condition on events of probability zero with sufficient regularity; see~\cite{PDCDSH}. 
One can reason a priori that if the conditioning event $E$ can be written as a random variable $T$ with a density evaluated at various values $k \in \mbox{supp}(T)$, then a similar PDC algorithm can be adapted, see~\cite{PDCDSH, DeSalvoCT}. 

In particular, like the decomposition of geometric random variables into their individual bits, an exponential distribution also has an analogous property. 
For real $x$, $\{x \}$ denotes the fractional part of $x$, and $\lfloor x \rfloor $ denotes the integer part of $x$, so that $x = \lfloor x \rfloor + \{x \}$. 

\begin{lemma}
Let $Y$ be an exponentially distributed random variable with parameter $\lambda>0$, then:
\begin{itemize}
\item the integer part, $\lfloor Y \rfloor$, and the fractional part, $\{Y\}$, are independent~\cite{Rejection, SteutelThiemann};
\item $\lfloor Y\rfloor$ is geometrically distributed with parameter $1-e^{-\lambda}$, and $\{Y\}$ has density $f_\lambda(x) = \lambda e^{-\lambda x} / (1-e^{-\lambda})$, $0\leq x < 1$. 
\end{itemize}
\end{lemma}

Using this property, a random sampling algorithm for nonnegative real-valued $(r,c)$-contingency tables is presented in~\cite[Algorithm~6]{DeSalvoCT}. 
The algorithm first samples the fractional part of each entry of the table; conditional on this first step, the remaining sampling problem is the usual random sampling of nonnegative integer-valued $(r',c')$-contingency table, for which Algorithm~\ref{CT:algorithm} is applicable. 

We are not aware of any non-trivial generalizations of Latin squares to real-valued entries. 
Here is one which may be of interest. 
Define a partition $J_1, \ldots, J_n$ of the interval $[0,n]$, and demand that that each row and each column of a matrix $M$ has exactly one entry in each of the sets $J_1, \ldots, J_n$. 
One can ask, then, to sample from the uniform measure over the set of all such matrices, which has a density with respect to Lebesgue measure. 
Another generalization would be to take $J_1, \ldots, J_n$ such that $J_1 \cup \ldots \cup J_n = [0,n]$, without the assumption that the sets form a partition, i.e., allow overlap. 
}


\ignore{

\section{Calculations for small $n$}
\label{section_calculations}
\subsection{$n=2$}

Consider the set of Latin squares of order~$2$, i.e., $\left(\begin{array}{cc}1&2\\2&1\end{array}\right)$ and $\left(\begin{array}{cc}2&1\\1&2\end{array}\right)$. 
These correspond to the bit-wise decomposition
\begin{align}
\left(\begin{array}{cc}1&2\\2&1\end{array}\right) & = 2^1 \left(\begin{array}{cc}0&1\\1&0\end{array}\right) + 2^0 \left(\begin{array}{cc}1&\blacksquare\\\blacksquare&1\end{array}\right) \\
\left(\begin{array}{cc}2&1\\1&2\end{array}\right) & = 2^1 \left(\begin{array}{cc}1&0\\0&1\end{array}\right) + 2^0 \left(\begin{array}{cc}\blacksquare&1\\ 1&\blacksquare\end{array}\right),
\end{align}
where we use $\blacksquare$ to denote a forced 0 entry. 
An application of Algorithm~\ref{latin:square:algorithm} starts by sampling from $X_{11}$, which to yield the uniform distribution over such tables must be $0$ or $1$ with equal probability. 
We have $c_1 = c_2 = 1$, $m=2$, $q_1 = q_2 = \frac{1}{3}$ and $\frac{q_1}{1+q_1} = \frac{q_2}{1+q_2} = 1/4$. 
The algorithm generates $X_{11}$ as a Bernoulli$(1/4)$ and applies one of the following rejection probabilities depending on whether the outcome is $0$ or $1$. 
\begin{align*}
R_0 = \P(X_{12} = 1) \P(X_{21}=1) \P(X_{22} = 0) = \frac{3^2}{4^2}, \\
R_1 = \P(X_{12} = 0) \P(X_{21}=0) \P(X_{22} = 1) = \frac{3}{4^2}. 
\end{align*}
The scaling factor in this case is given by $\alpha = \min(R_0, R_1) = \frac{3}{4^2}$. 
Let $\tilde R_0 = R_0 / \alpha$, and $\tilde R_1 = R_1 / \alpha$. 
For this particular example we have $\tilde R_0 = 1/3$. 
The overall probability of accepting $0$ for the $(1,1)$ entry is then given by the geometric series 
\begin{align*}
 \P(\tilde X_{11}=0) & = \P(X_{11} = 0) \tilde R_0 + \P(X_{11}=0)^2 (1-\tilde R_0) \tilde R_0 + \P(X_{11}=0)^3 (1-\tilde R_0)^2 \tilde R_0 + \cdots \\
     & = \frac{ \P(X_{11}=0)\tilde R_0}{1-\P(X_{11}=0)(1-R_0)} = \frac{1}{2}. 
\end{align*}
Thus, each outcome is produced with probability $1/2$, corresponding to the uniform distribution. 

\subsection{$n=3$}

Rather than apply Algorithm~\ref{latin:square:algorithm} directly to the case $n=3$, we shall approach it in the opposite direction by looking at various $3 \times 3$ tables. 

Let us take for our next example the $3\times 3$ table with row sums and column sums all equal to 1, of which there are $6!$ tables corresponding to the permutation matrices of $3$. 
This case allows us an easy way to asses potential bias in Algorithm~\ref{binary:algorithm:approx}. 
In this case we have $c_1 = 1$, $m=3$, $q_1 = 1/4$, and $\frac{q_1}{1+q_1} = 1/5$. 
We have 
\begin{align*}
 R_0 & = \P(X_{12} +X_{13} = 1) \P(X_{21}+X_{31}=1) = \left(2\, \frac{4}{5^2}\right)^2 = \frac{4^3}{5^4}. \\
 R_1 & = \P(X_{12} +X_{13} = 0) \P(X_{21}+X_{31}=0) = \left(\frac{4^2}{5^2}\right)^2 = \frac{4^4}{5^4}. 
\end{align*}
Hence, $\tilde R_0 = 1/4$, and as noted previously we must have $\tilde R_1 = 1$. 
The overall probability of accepting a $0$ in entry $(1,1)$ is therefore $1/2$, by the same geometric series argument, whereas for the uniform distribution over permutation matrices of size~$3$ it should be $2/3.$
We thus tend to under sample permutation matrices with a $1$ in entry $(1,1)$, with an error proportion not exceeding $|1/2 - 2/3| = 1/6$. 

Next, we consider a $3 \times 3$ table with row sums and column sums all equal to $1$, with the following fixed zeros
\[ \left(\begin{array}{ccc}  \underline{\ \ }& \blacksquare &\underline{\ \ } \\ \blacksquare& \underline{\ \ }&\underline{\ \ } \\ \underline{\ \ }&\underline{\ \ } & \blacksquare\end{array}\right). \]
WLOG, this covers the case of all possible $3\times 3$ tables row sums and column sums all equal to 1 with exactly one entry in each row and column fixed to be 0. 
We have $c_1 = 1, m=3, h_1 = 1$, and so $q_1 = \frac{c_1}{c_1+m-h_1}$
}

\ignore{
\section{Further directions}

It is most common when working with Latin squares to normalize the first row and the first column to be the identity permutation $\{1,2,\ldots,n\}$. 
We have not carried out this in the interest of simplicity, but this would be a natural optimization to the current approach. 
}

\ignore{
\section{Recycle}

Our results are now summarized below.

\begin{theorem}\label{exact:theorem:tables}
Let $m$ denote the number of rows and $n$ denote the number of columns in a nonnegative integer-valued $(r,c)$-contingency table, and let $M$ denote the largest row sum or column sum. 
Let $W$ denote an $m\times n$ matrix with a $1$ in the $(i,j)$th entry if the $(i,j)$-th entry of the contingency table is forced to be 0, and let $s := m\, n - \sum_{i,j} W_{i,j}$ denote the number of entries not forced to be 0. 
\begin{enumerate}
\item[(1)] Algorithm~\ref{CT:algorithm} returns a nonnegative integer-valued $(r,c)$-contingency table, uniform over all such tables. 
The expected number of random bits required is $O( s \log M)$. 
\item[(2)] Algorithm~\ref{CT:algorithm:approx} returns a nonnegative integer-valued $(r,c)$-contingency table, which is approximately uniform. 
The expected number of random bits required is $O( s\, m\, \log M)$, with $O(s\, n\, M\, \log^2 M)$ arithmetic operations. 
\end{enumerate}
\end{theorem}

Item (1)~in Theorem~\ref{exact:theorem:tables} demonstrates that, should one be able to compute certain numerical quantities efficiently, one would have an optimal sampling method.  
Item (2)~in Theorem~\ref{exact:theorem:tables} offers an alternative explicit and practical approximation algorithm which does not require any enumeration formulas, and instead uses convolution of independent random variables with explicitly known distribution functions. 
Analogous results hold for binary $(r,c)$-contingency tables. 

In addition, motivated by the initial results concerning Shannon's entropy of the collections of Latin squares and Sudoku matrices in~\cite{DeSalvoNewton}, later extended in~\cite{Cameron1, Cameron2}, the goal of sampling uniformly from Latin squares and Sudoku matrices allows us to glean important information about various statistics of interest in information theory, e.g., Shannon's entropy, compared to uniformly random matrices without restrictions. 
Several PDC algorithms were applied recently for the \emph{exact} sampling from Latin squares and Sudoku matrices in~\cite{DeSalvoSudoku}, by taking advantage of certain statistics in, e.g.,~\cite{FelgenJarvis, Sudoku2}, extending the range of practical exact sampling methods for these structures from those in~\cite{Dahl, Yordzhev2}. 
There are, however, alternatives like importance sampling for Sudoku matrices~\cite{RidderIS}, and also Markov chain techniques which have been specifically applied to Latin squares, see for example~\cite{Fontana, FontanaFractions}.  

\begin{theorem}\label{latin:square:theorem}
Algorithm~\ref{latin:square:algorithm} produces a valid Latin square of order $n$, approximately uniform over all such Latin squares, which terminates in finite time a.s. 
The expected number of random bits required is $O(n^3\, \log(n))$, with $O(n^4 \log n)$ arithmetic operations.  
\end{theorem}

One can also generalize the aforementioned tables to having continuous--valued entries, in which case many of the traditional sampling algorithms either break down or require significant adaptation. 
Using PDC, such generalities are often easily handled, often by adding in just one extra step, see Section~\ref{continuous}; see also~\cite[Algorithm 6]{DeSalvoCT} and~\cite{PDCDSH}.

We start with a well-known probabilistic model for the entries in a random contingency table, generalized to include entries which are forced to be 0.
In this section, we let $W$ be any given $m\times n$ matrix with values in $\{0,1\}$. 
For each $j=1,2,\ldots,n$, we define $J_j := \{1,\ldots, n\} \setminus \{ i : W(i,j) = 1\}$, and let $k_j$ denote the last entry in $J_j$. 
Similarly, for each $i=1,2,\ldots,m$, we define $I_i := \{1,\ldots, n\} \setminus \{j : W(i,j) = 1\}$, and let $\ell_i$ denote the last entry in $I_i$, $i=1,2,\ldots,m$. 
Also, we let $h_j = \sum_{i=1}^m W(i,j)$ denote the number of entries forced to be zero in column $j$, for $j=1,2,\ldots,n$.

\begin{lemma}
\label{lemma:uniform}
Let $\X_W=(X_{ij})_{1 \leq i \leq m, 1 \leq j \leq n}$ denote a collection of independent geometric random variables with parameters $p_{ij}$, such that $W(i,j) = 1$ implies $X_{i,j}$ is conditioned to have value 0.   
If $p_{ij}$ has the form $p_{ij} = 1 - \alpha_i \beta_j$, then $\X_W$ is uniform restricted to $(r,c)$-contingency tables with zeros in entries indicated by $W$.
\end{lemma}
\begin{proof}
Let $\X=(X_{ij})_{1 \leq i \leq m, 1 \leq j \leq n}$ denote a collection of independent geometric random variables with parameters $p_{ij}$, where $p_{ij}$ has the form $p_{ij} = 1 - \alpha_i \beta_j$. 
We have 
\[\P\big(\X=\xi\big)
=\prod_{i,j}\P\big(X_{ij}=\xi_{ij}\big)
=\prod_{i,j}(\alpha_i \beta_j)^{\xi_{ij}}(1-\alpha_i \beta_j)
=\prod_i\alpha_i^{r_i} \prod_j\beta_j^{c_j}\prod_{i,j}\big(1-\alpha_i \beta_j\big).
\]
Since this probability does not depend on $\xi$, it follows that the restriction of $\X$ to $(r,c)$-contingency tables is uniform. 
As the collection of random variables are independent, conditioning on any $X_{i,j} = 0$ only changes the constant of proportionality, and does not affect the dependence on the $\xi$, hence 
\[\P\big(\X_W=\xi\big)
=\prod_{i,j : W(i,j) = 0}\P\big(X_{ij}=\xi_{ij}\big)
=\prod_{i,j: W(i,j)=0}(\alpha_i \beta_j)^{\xi_{ij}}(1-\alpha_i \beta_j)
=\prod_i\alpha_i^{r_i} \prod_j\beta_j^{c_j}\prod_{i,j}\big(1-\alpha_i \beta_j\big);
\]
i.e., it follows that the restriction of $\X_W$ to $(r,c)$-contingency tables with forced zero entries indicated by $W$ is uniform. 
\end{proof}

\begin{lemma}\label{expectation}
Suppose $\X_W$ is a collection of independent geometric random variables, where $X_{i,j}$ has parameter $p_{ij} = \frac{m-h_j}{m-h_j+c_j}$, for all pairs $(i,j)$ such that $W(i,j) = 0$, and $p_{i,j} = 1$ for all pairs $(i,j)$ such that $W(i,j) = 1$.  Then the expected column sums of $\X_W$ are $c_1, c_2, \ldots, c_n$, and the expected row sums are $\sum_{j\in I_1} \frac{c_j}{m-h_j}, \ldots, \sum_{j\in I_m} \frac{c_j}{m-h_j}$. 
\end{lemma}
\begin{proof}
Note first that $\E X_{i,j} = p_{i,j}^{-1} - 1$.  Then 
for any $j=1,2,\ldots,n$,
\[ \sum_{i=1}^m \E\, X_{i,j} = \sum_{i \in J_j} \frac{m-h_j+c_j}{m-h_j} - 1= c_j  \]
and similarly for any $i=1,2,\ldots, m$, 
\[ \sum_{j=1}^n \E\, X_{i,j} = \sum_{j \in I_i} \frac{m-h_j+c_j}{m-h_j}-1 = \sum_{j\in I_i} \frac{c_j}{m-h_j}. \qedhere\]
\end{proof}

The proof of uniformity of Algorithm~\ref{CT:algorithm} now follows almost verbatim from~\cite[Section~4]{DeSalvoCT}, aside from straightforward modifications of boundary conditions and summing over index sets $j \in I_i$ rather than $j=1,2,\ldots,n$. 
Start by considering any the first entry in the first column not forced to be 0, say at entry $(s,1)$, and denote the least significant bit by $\epsilon_{s,1}$. 
As we are in the first column, we have $\L(\epsilon_{s,1}) = \L\left(\Bern\left(\frac{q_1}{1+q_1}\right)\right)$, and we reject according to the correct proportion $\P(E | \epsilon_{s,1})$. 
In fact, 
as we reject \emph{in proportion} to this probability, we normalize by all terms which do not depend on $k$, which gives
\[ \P(E | \epsilon_{s,1} = k) \propto \Sigma\left(\begin{array}{l} (\ldots,r_s-k,\ldots, r_{i_1} - b_{1}(k),\ldots, r_{i_r} - b_{r}(k)), \\ (\ldots,c_1-k, \ldots, c_{j_1} - b_{1}(k),\ldots, c_{j_r} - b_{r}(k)), \\ W,\ \mathcal{O}_{i,j};\  t_1, \ldots, t_r \end{array} \right)\,\cdot \, q_1^{-k}\, \left(\prod_{\ell=1}^r q_{j_\ell}^{-b_{\ell}(k)}\right).\]
Since $\P(\epsilon_{s,1} = k) \propto q_1^k$, the rejection function in Equation~\eqref{equation:mn} follows. 
The rest of the proof follows by induction in a straightforward manner, cf.~\cite[Section~4]{DeSalvoCT}. 

The second part of the theorem, with rejection function given in Equation~\eqref{alternative:rejection}, is motivated by the following alternative probabilistic formulation of Equation~\eqref{equation:mn}:
 \begin{align}\label{fij1}
\hspace{-0.25in} f(i,j,k, r, c,W)  \propto \ &
 \P\left( \begin{array}{llll}  
\sum_{\ell\in J_1}^{\ell < j} 2\xi''_{1,\ell}(q_\ell^2, c_\ell) &+ \eta_{1,j,i}'(q_j, c_j) &+  \sum_{\ell\in J_1}^{\ell > j} \xi'_{1,j}(q_\ell, c_\ell) &= r_1 \\
\sum_{\ell\in J_2}^{\ell < j} 2\xi''_{2,\ell}(q_\ell^2, c_\ell) &+ \eta_{2,j,i}'(q_j, c_j) &+  \sum_{\ell\in J_2}^{\ell > j} \xi'_{2,j}(q_\ell, c_\ell) &= r_2 \\
\qquad \vdots \\
\sum_{\ell\in J_{i-1}}^{\ell < j} 2\xi''_{i-1,\ell}(q_\ell^2, c_\ell) &+ \eta_{i-1,j,i}'(q_j, c_j) &+  \sum_{\ell\in J_{i-1}}^{\ell > j} \xi'_{i-1,j}(q_\ell, c_\ell) &= r_{i-1} \\
\sum_{\ell\in J_i}^{\ell < j}2\xi''_{i,\ell}(q_\ell^2, c_\ell) &+ \eta_{i,j,i}''(q_j, c_j) &+  \sum_{\ell\in J_i}^{\ell > j} \xi'_{i,j}(q_\ell, c_\ell) &= r_{i}-k \\
\sum_{\ell\in J_{i+1}}^{\ell < j}2\xi''_{i+1,\ell}(q_\ell^2, c_\ell) &+ \eta_{i+1,j,i}''(q_j, c_j) &+  \sum_{\ell\in J_{i+1}}^{\ell > j} \xi'_{i+1,j}(q_\ell, c_\ell) &= r_{i+1} \\
\qquad \vdots \\
\sum_{\ell\in J_m}^{\ell < j}2\xi''_{m,\ell}(q_\ell^2, c_\ell) &+ \eta_{m,j,i}''(q_j, c_j) &+  \sum_{\ell\in J_m}^{\ell > j} \xi'_{m,j}(q_\ell, c_\ell) &= r_{m} \\
\end{array}\right) \\
\nonumber & \ \ \ \  \times  \P\left( \sum_{\ell\in I_i, \ell \leq i} 2\, \xi_{\ell,j}(q_j^2) + \sum_{\ell\in I_i, \ell > i} \xi_{\ell,j}(q_\ell) = c_j - k \right). 
\end{align} 
If we could evaluate the probability above exactly, or to some arbitrarily defined precision, then we would obtain an exact sampling algorithm. 
However, we surmise that the dependencies between the random variables are strongest along the $i$-th row and the $j$-th column, which is why we champion the rejection function in Equation~\eqref{alternative:rejection}, as it captures what is most likely the dominant source of bias, and enforces the parity condition as well. 
}

\section{Acknowledgements}
The author would like to acknowledge helpful discussions with Alejandro Morales, Igor Pak, Richard Arratia, and James Zhao.

\bibliographystyle{plain}
\bibliography{../../../master_bib}

\begin{thebibliography}{10}

\bibitem{PDC}
Richard Arratia and Stephen DeSalvo.
\newblock Probabilistic divide-and-conquer: a new exact simulation method, with
  integer partitions as an example.
\newblock {\em Combinatorics, Probability and Computing}, 25(3):324--351, May
  2016.

\bibitem{BaldoniSilva}
W.~Baldoni-Silva, J.~A. De~Loera, and M.~Vergne.
\newblock Counting integer flows in networks.
\newblock {\em Found. Comput. Math.}, 4(3):277--314, 2004.

\bibitem{BarvinokInequalities}
Alexander Barvinok.
\newblock Brunn-{M}inkowski inequalities for contingency tables and integer
  flows.
\newblock {\em Adv. Math.}, 211(1):105--122, 2007.

\bibitem{BarvinokPermanents}
Alexander Barvinok.
\newblock Enumerating contingency tables via random permanents.
\newblock {\em Combin. Probab. Comput.}, 17(1):1--19, 2008.

\bibitem{Barvinok}
Alexander Barvinok.
\newblock Asymptotic estimates for the number of contingency tables, integer
  flows, and volumes of transportation polytopes.
\newblock {\em International Mathematics Research Notices}, 2009(2):348--385,
  2009.

\bibitem{BarvinokIntegerFlows}
Alexander Barvinok.
\newblock Asymptotic estimates for the number of contingency tables, integer
  flows, and volumes of transportation polytopes.
\newblock {\em Int. Math. Res. Not. IMRN}, (2):348--385, 2009.

\bibitem{BarvinokRandomCT}
Alexander Barvinok.
\newblock What does a random contingency table look like?
\newblock {\em Combinatorics, Probability and Computing}, 19(04):517--539,
  2010.

\bibitem{BarvinokHartigan}
Alexander Barvinok and J.~A. Hartigan.
\newblock An asymptotic formula for the number of non-negative integer matrices
  with prescribed row and column sums.
\newblock {\em Trans. Amer. Math. Soc.}, 364(8):4323--4368, 2012.

\bibitem{BarvinokApproximate}
Alexander Barvinok, Zur Luria, Alex Samorodnitsky, and Alexander Yong.
\newblock An approximation algorithm for counting contingency tables.
\newblock {\em Random Structures Algorithms}, 37(1):25--66, 2010.

\bibitem{BenderTables}
Edward~A. Bender.
\newblock The asymptotic number of non-negative integer matrices with given row
  and column sums.
\newblock {\em Discrete Math.}, 10:217--223, 1974.

\bibitem{bender1978asymptotic}
Edward~A Bender and E~Rodney Canfield.
\newblock The asymptotic number of labeled graphs with given degree sequences.
\newblock {\em Journal of Combinatorial Theory, Series A}, 24(3):296--307,
  1978.

\bibitem{bezakova}
Ivona Bez{\'a}kov{\'a}, Nayantara Bhatnagar, and Eric Vigoda.
\newblock Sampling binary contingency tables with a greedy start.
\newblock {\em Random Structures \& Algorithms}, 30(1-2):168--205, 2007.

\bibitem{BrualdiDahl}
Richard~A. Brualdi and Geir Dahl.
\newblock Matrices of zeros and ones with given line sums and a zero block.
\newblock {\em Linear Algebra Appl.}, 371:191--207, 2003.

\bibitem{Brualdi2007}
Richard~A Brualdi and Geir Dahl.
\newblock Constructing (0, 1)-matrices with given line sums and certain fixed
  zeros.
\newblock In {\em Advances in Discrete Tomography and Its Applications}, pages
  113--123. Springer, 2007.

\bibitem{chen}
Yuguo Chen, Persi Diaconis, Susan~P Holmes, and Jun~S Liu.
\newblock Sequential monte carlo methods for statistical analysis of tables.
\newblock {\em Journal of the American Statistical Association},
  100(469):109--120, 2005.

\bibitem{chung1996sampling}
Fan R.~K. Chung, Ronald~L. Graham, and Shing-Tung Yau.
\newblock On sampling with markov chains.
\newblock {\em Random Structures and Algorithms}, 9(1-2):55--77, 1996.

\bibitem{colbourn1984complexity}
Charles~J Colbourn.
\newblock The complexity of completing partial latin squares.
\newblock {\em Discrete Applied Mathematics}, 8(1):25--30, 1984.

\bibitem{cryan2003polynomial}
Mary Cryan and Martin Dyer.
\newblock A polynomial-time algorithm to approximately count contingency tables
  when the number of rows is constant.
\newblock {\em Journal of Computer and System Sciences}, 67(2):291--310, 2003.

\bibitem{cryan2006rapidly}
Mary Cryan, Martin Dyer, Leslie~Ann Goldberg, Mark Jerrum, and Russell Martin.
\newblock Rapidly mixing markov chains for sampling contingency tables with a
  constant number of rows.
\newblock {\em SIAM Journal on Computing}, 36(1):247--278, 2006.

\bibitem{Dahl}
Geir Dahl.
\newblock Permutation matrices related to sudoku.
\newblock {\em Linear Algebra and its Applications}, 430(8):2457--2463, 2009.

\bibitem{DeLoeraRational}
Jes{{\'u}}s~A. De~Loera, Raymond Hemmecke, Jeremiah Tauzer, and Ruriko Yoshida.
\newblock Effective lattice point counting in rational convex polytopes.
\newblock {\em J. Symbolic Comput.}, 38(4):1273--1302, 2004.

\bibitem{LD}
Amir Dembo, Anatoly Vershik, and Ofer Zeitouni.
\newblock Large deviations for integer partitions, 1998.

\bibitem{PDCDSH}
Stephen DeSalvo.
\newblock Probabilistic divide-and-conquer: deterministic second half.
\newblock {\em arXiv preprint arXiv:1411.6698}, 2014.

\bibitem{DeSalvoSudoku}
Stephen DeSalvo.
\newblock Exact, uniform sampling of latin squares and sudoku matrices.
\newblock {\em Algorithmica}, To appear.

\bibitem{DeSalvoImprovements}
Stephen DeSalvo.
\newblock Improvements to exact {B}oltzmann sampling using probabilistic
  divide-and-conquer and the recursive method.
\newblock {\em Electronic Notes in Discrete Mathematics}, To appear.

\bibitem{DeSalvoCT}
Stephen DeSalvo and James~Y Zhao.
\newblock Random sampling of contingency tables via probabilistic
  divide-and-conquer.
\newblock {\em arXiv preprint arXiv:1507.00070}, 2015.

\bibitem{DiaconisGangolli}
Persi Diaconis and Anil Gangolli.
\newblock Rectangular arrays with fixed margins.
\newblock In {\em Discrete probability and algorithms ({M}inneapolis, {MN},
  1993)}, volume~72 of {\em IMA Vol. Math. Appl.}, pages 15--41. Springer, New
  York, 1995.

\bibitem{diaconis1993comparison}
Persi Diaconis and Laurent Saloff-Coste.
\newblock Comparison theorems for reversible markov chains.
\newblock {\em The Annals of Applied Probability}, pages 696--730, 1993.

\bibitem{diaconissturmfels}
Persi Diaconis, Bernd Sturmfels, et~al.
\newblock Algebraic algorithms for sampling from conditional distributions.
\newblock {\em The Annals of statistics}, 26(1):363--397, 1998.

\bibitem{Duchon:2011aa}
Philippe Duchon.
\newblock Random generation of combinatorial structures: Boltzmann samplers and
  beyond.
\newblock 12 2011.

\bibitem{Boltzmann}
Philippe Duchon, Philippe Flajolet, Guy Louchard, and Gilles Schaeffer.
\newblock Boltzmann samplers for the random generation of combinatorial
  structures.
\newblock {\em Combin. Probab. Comput.}, 13(4-5):577--625, 2004.

\bibitem{dyergreenhill}
Martin Dyer and Catherine Greenhill.
\newblock Polynomial-time counting and sampling of two-rowed contingency
  tables.
\newblock {\em Theoretical Computer Science}, 246(1):265--278, 2000.

\bibitem{dyer1997sampling}
Martin Dyer, Ravi Kannan, and John Mount.
\newblock Sampling contingency tables.
\newblock {\em Random Structures and Algorithms}, 10(4):487--506, 1997.

\bibitem{PoissonBinomial}
Manuel Fernandez and Stuart Williams.
\newblock Closed-form expression for the {P}oisson-binomial probability density
  function.
\newblock {\em Aerospace and Electronic Systems, IEEE Transactions on},
  46(2):803--817, 2010.

\bibitem{fishman2012counting}
George~S Fishman.
\newblock Counting contingency tables via multistage markov chain monte carlo.
\newblock {\em Journal of Computational and Graphical Statistics},
  21(3):713--738, 2012.

\bibitem{Flajolet}
Philippe Flajolet and Robert Sedgewick.
\newblock {\em Analytic combinatorics}.
\newblock Cambridge University Press, Cambridge, 2009.

\bibitem{FontanaFractions}
Roberto Fontana.
\newblock Fractions of permutations. an application to sudoku.
\newblock {\em Journal of Statistical Planning and Inference},
  141(12):3697--3704, 2011.

\bibitem{Fontana}
Roberto Fontana.
\newblock Random latin squares and sudoku designs generation.
\newblock {\em arXiv preprint arXiv:1305.3697}, 2013.

\bibitem{Fristedt}
Bert Fristedt.
\newblock The structure of random partitions of large integers.
\newblock {\em Transactions of the American Mathematical Society},
  337(2):703--735, 1993.

\bibitem{GoodCrook}
I.~J. Good and J.~F. Crook.
\newblock The enumeration of arrays and a generalization related to contingency
  tables.
\newblock {\em Discrete Math.}, 19(1):23--45, 1977.

\bibitem{GreenhillMcKay}
Catherine Greenhill and Brendan~D. McKay.
\newblock Asymptotic enumeration of sparse nonnegative integer matrices with
  specified row and column sums.
\newblock {\em Adv. in Appl. Math.}, 41(4):459--481, 2008.

\bibitem{hernek1998random}
Diane Hernek.
\newblock Random generation of 2$\times$ n contingency tables.
\newblock {\em Random Structures \&amp; Algorithms}, 13(1):71--79, 1998.

\bibitem{hua1942number}
Loo-keng Hua.
\newblock On the number of partitions of a number into unequal parts.
\newblock {\em Trans. Amer. Math. Soc.}, 51:194--201, 1942.

\bibitem{huber2015perfect}
Mark~L. Huber.
\newblock {\em Perfect Simulation}.
\newblock Chapman \& Hall/CRC Monographs on Statistics \& Applied Probability.
  Taylor \& Francis, 2015.

\bibitem{JacobsonMatthews}
Mark~T Jacobson and Peter Matthews.
\newblock Generating uniformly distributed random latin squares.
\newblock {\em Journal of Combinatorial Designs}, 4(6):405--437, 1996.

\bibitem{JerrumSinclair}
Mark Jerrum, Alistair Sinclair, and Eric Vigoda.
\newblock A polynomial-time approximation algorithm for the permanent of a
  matrix with nonnegative entries.
\newblock {\em J. ACM}, 51(4):671--697 (electronic), 2004.

\bibitem{VershikKerov}
Sergei~V. Kerov and Anatol~M. Vershik.
\newblock The characters of the infinite symmetric group and probability
  properties of the {R}obinson-{S}chensted-{K}nuth algorithm.
\newblock {\em SIAM J. Algebraic Discrete Methods}, 7(1):116--124, 1986.

\bibitem{kitajimamatsui}
Shuji Kijima and Tomomi Matsui.
\newblock Polynomial time perfect sampling algorithm for two-rowed contingency
  tables.
\newblock {\em Random Structures \& Algorithms}, 29(2):243--256, 2006.

\bibitem{morris2002improved}
Ben~J Morris.
\newblock Improved bounds for sampling contingency tables.
\newblock {\em Random Structures \&amp; Algorithms}, 21(2):135--146, 2002.

\bibitem{DeSalvoNewton}
Paul~K Newton and Stephen~A DeSalvo.
\newblock The shannon entropy of sudoku matrices.
\newblock In {\em Proceedings of the Royal Society of London A: Mathematical,
  Physical and Engineering Sciences}, page rspa20090522. The Royal Society,
  2010.

\bibitem{NW}
Albert Nijenhuis and Herbert~S. Wilf.
\newblock A method and two algorithms on the theory of partitions.
\newblock {\em J. Combinatorial Theory Ser. A}, 18:219--222, 1975.

\bibitem{NWbook}
Albert Nijenhuis and Herbert~S. Wilf.
\newblock {\em Combinatorial algorithms}.
\newblock Academic Press, Inc. [Harcourt Brace Jovanovich, Publishers], New
  York-London, second edition, 1978.
\newblock For computers and calculators, Computer Science and Applied
  Mathematics.

\bibitem{ONeil}
Patrick~Eugene O'Neil.
\newblock Asymptotics and random matrices with row-sum and column-sum
  restrictions.
\newblock {\em Bull. Amer. Math. Soc.}, 75:1276--1282, 1969.

\bibitem{PittelSetPartitions}
Boris Pittel.
\newblock Random set partitions: asymptotics of subset counts.
\newblock {\em journal of combinatorial theory, Series A}, 79(2):326--359,
  1997.

\bibitem{ProppWilson}
James~Gary Propp and David~Bruce Wilson.
\newblock Exact sampling with coupled markov chains and applications to
  statistical mechanics.
\newblock {\em Random structures and Algorithms}, 9(1-2):223--252, 1996.

\bibitem{Soules}
George~W. Soules.
\newblock New permanental upper bounds for nonnegative matrices.
\newblock {\em Linear Multilinear Algebra}, 51(4):319--337, 2003.

\bibitem{vanLintWilson}
J.~H. van Lint and R.~M. Wilson.
\newblock {\em A course in combinatorics}.
\newblock Cambridge University Press, Cambridge, second edition, 2001.

\bibitem{Vershik}
A.M. Vershik.
\newblock Statistical mechanics of combinatorial partitions, and their limit
  shapes.
\newblock {\em Functional Analysis and Its Applications}, 30:90--105, 1996.

\bibitem{wilkinson2006parallel}
Darren~J Wilkinson.
\newblock Parallel bayesian computation.
\newblock {\em Statistics Textbooks and Monographs}, 184:477, 2006.

\bibitem{YordzhevNumber}
Krasimir Yordzhev.
\newblock On the number of disjoint pairs of s-permutation matrices.
\newblock {\em Discrete Applied Mathematics}, 161(18):3072--3079, 2013.

\end{thebibliography}

\end{document}